\documentclass[11pt]{article}
\usepackage{color,soul}
\usepackage{titlesec}
\usepackage{subcaption}
\usepackage{breakcites}
\usepackage[pagebackref]{hyperref}
\usepackage{dirtytalk}
\usepackage{lipsum}
\usepackage{ntheorem}
\usepackage{mathtools}
\usepackage{calc}
\usepackage{stmaryrd}
\usepackage{paralist}
\newlength{\depthofsumsign}
\usepackage[figuresright]{rotating}
\usepackage{epsfig}
\usepackage{curves}
\usepackage{graphicx}

\usepackage[T1]{fontenc}
\usepackage[utf8]{inputenc}

\usepackage{algorithm,algpseudocode}

\usepackage{graphicx}
\usepackage{amsmath}
\usepackage{amssymb}
\usepackage{amsfonts}
\usepackage{array}
\usepackage{tikz}
\usepackage{hhline}
\usepackage{multirow}
\usepackage{hyperref}
\usepackage{mathtools}
\usepackage{tikz}
\usepackage{makecell}
\usepackage{pgfplots}
\usepackage{afterpage}
\usepackage{fancyhdr}
\usepackage{setspace}
\usepackage{bm}
\usepackage{authblk}

\newtheorem{proposition}{Proposition}
\newtheorem*{proposition-non}{Proposition}

\newtheorem{theorem}{Theorem}

\newtheorem{conjecture}{Conjecture}
\newtheorem{remark}{Remark}

\allowdisplaybreaks

\newenvironment{proof}{\noindent{\bf Proof:}\indent}%
                      {\hfill $\Box$\par}

\newcommand{\sym}[1]{{\sf #1}}

\title{The Grain Family of Stream Ciphers: an Abstraction, Strengthening of Components and New Concrete Instantiations}
\author[ ]{Palash Sarkar}
\affil[ ]{Indian Statistical Institute, 203, B.T. Road, Kolkata, India 700108}
\affil[ ]{Email: {\tt palash@isical.ac.in}}

\date{\today}

\begin{document}

\maketitle

\begin{abstract}
	The first contribution of the paper is to put forward an abstract definition of the Grain family of stream ciphers which formalises the different
	components that are required to specify a
	particular member of the family. Our second contribution is to provide new and strengthened definitions of the components. These include definining new 
	classes of nonlinear Boolean functions, improved definition of the state update function during initialisation, choice of the tap positions, and the
	possibility of the linear feedback shift register being smaller than the nonlinear feedback shift register. The third contribution
	of the paper is to put forward seven concrete proposals of stream ciphers by suitably instantiating the abstract family, one at the 80-bit security level, and 
	two each at the 128-bit, 192-bit, and the 256-bit security levels. At the 80-bit security level, compared to the well known Grain~v1, the new proposal uses
	Boolean functions with improved cryptographic properties \textit{and} an overall lower gate count. At the 128-bit level, compared to ISO/IEC standard 
	Grain-128a, the new proposals use Boolean functions with improved cryptographic properties; one of the proposals require a few extra gates, while the other
	has an overall lower gate count. At the 192-bit, and the 256-bit security levels, there are no proposals in the literature with smaller gate counts. 
\end{abstract}

\section{Introduction \label{sec-intro}}
Stream ciphers and block ciphers are the two basic methods of performing symmetric key encryption. Constructions of stream ciphers are typically geared either towards small 
hardware footprint, or towards very fast software implementation. The focus of the present work is on the former type of stream ciphers. Over the last two decades several 
hardware oriented stream ciphers have been proposed in the literature. One of the most successful approaches to such stream ciphers is the Grain family of stream 
ciphers~\cite{Grain-org,DBLP:journals/ijwmc/HellJM07,DBLP:conf/isit/Hell0MM06,DBLP:journals/ijwmc/AgrenHJM11,Grain-128AEADv2-Eprint,Grain-128AEADv2}.
Other well known constructions include MICKEY 2.0~\cite{DBLP:series/lncs/BabbageD08}, Trivium~\cite{DBLP:series/lncs/CanniereP08},
Fruit-80~\cite{DBLP:journals/entropy/GhafariH18}, Sprout~\cite{DBLP:conf/fse/ArmknechtM15-sprout} and its successor Plantlet~\cite{DBLP:journals/tosc/MikhalevAM16},
and Lizard~\cite{DBLP:journals/tosc/HamannKM17}. The regular appearance of new designs is testimony to the enduring interest of the stream cipher community in
new proposals of hardware oriented stream ciphers. 

The approach taken by the Grain family of stream ciphers is to combine a linear feedback shift register (LFSR) and a nonlinear feedback shift register (NFSR). 
The stream cipher takes as input a secret key and a publicly known initialisation vector (IV) and produces as output a sequence of bits called the keystream. Typically 
the keystream bit sequence is XORed with the message sequence to perform encryption. Some versions of the Grain family also provide an option for performing authenticated 
encryption with associated data (AEAD), while some other versions perform only AEAD and do not have the option of performing only encryption.

Several variants of Grain have been proposed. The first version~\cite{Grain-org} was targeted towards the 80-bit security level, but was 
cryptanalysed~\cite{DBLP:conf/fse/BerbainGM06} soon after being proposed, and led to a modified design named Grain~v1~\cite{DBLP:journals/ijwmc/HellJM07}. At 
the 128-bit security level, a proposal named Grain-128 was put forward~\cite{DBLP:conf/isit/Hell0MM06}. Grain-128 was cryptanalysed using dynamic cube 
attacks~\cite{DBLP:conf/fse/DinurS11a,DBLP:conf/asiacrypt/DinurGPSZ11}, and subsequently a modified version called Grain-128a was proposed~\cite{DBLP:journals/ijwmc/AgrenHJM11}. 
Both Grain~v1 and Grain-128a perform encryption, while Grain-128a optionally performs authenticated encryption. In response to a call for proposals for light weight
designs by the NIST of USA, a modification of Grain-128a, named Grain-128AEAD, was put forward, and was later modified to the
construction Grain-128AEADv2~\cite{Grain-128AEADv2-Eprint,Grain-128AEADv2}. Grain-128AEADv2 performs authenticated encryption with associated data and does not have
an encryption only mode of operation. 
Grain~v1 (along with MICKEY 2.0~\cite{DBLP:series/lncs/BabbageD08} and Trivium~\cite{DBLP:series/lncs/CanniereP08}) was included in the eSTREAM final portfolio of stream cipher 
algorithms in the hardware category. Grain-128a is an ISO/IEC standard (ISO/IEC 29167-13:2015).

A number of papers have studied various aspects of the security of the Grain family of stream ciphers. These attacks include
linear approximations and fast correlation attacks~\cite{DBLP:conf/fse/BerbainGM06,DBLP:conf/ccs/Maximov06,DBLP:conf/crypto/TodoIMAZ18},
slide attacks~\cite{Kucuk2006,DBLP:conf/africacrypt/CanniereKP08},
cube attacks~\cite{DBLP:conf/fse/DinurS11a,DBLP:conf/asiacrypt/DinurGPSZ11,DBLP:conf/crypto/TodoIHM17,DBLP:conf/crypto/WangHTLIM18,DBLP:conf/eurocrypt/HaoLMT020},
near collision attacks~\cite{DBLP:conf/fse/ZhangLFL13,DBLP:conf/eurocrypt/ZhangXM18,DBLP:conf/acns/BanikCM23},
conditional differential cryptanalysis~\cite{DBLP:conf/asiacrypt/KnellwolfMN10,DBLP:conf/cans/LehmannM12,DBLP:conf/acisp/Banik14,DBLP:journals/ccds/Banik16,DBLP:journals/iet-ifs/MaTQ17,DBLP:journals/iet-ifs/MaTQ17a,DBLP:journals/iet-ifs/LiG19},
time/memory trade-off attacks~\cite{DBLP:journals/dcc/DalaiPS22,DBLP:journals/tit/KumarS23},
related key attacks~\cite{DBLP:conf/acisp/LeeJSH08,DBLP:conf/acisp/BanikMST13,DBLP:journals/tifs/DingG13},
and fault attacks~\cite{DBLP:conf/host/CastagnosBCDGGPS09,DBLP:conf/africacrypt/KarmakarC11,DBLP:conf/ches/BanikMS12,DBLP:conf/date/DeyCAM15,DBLP:conf/indocrypt/BanikMS12,DBLP:journals/tc/SarkarBM15,DBLP:journals/ccds/MazumdarMS15,DBLP:conf/space/SiddhantiSMC17,DBLP:journals/chinaf/LiLL24a,DBLP:journals/tcad/ChakrabortyMM17,DBLP:conf/space/BanikMS12a,DBLP:journals/access/SalamOXYPP21}.

Many such works have studied reduced versions of Grain,
where the number of initialisation rounds is less than the specified number of rounds. The only known attack on the ``full'' versions of Grain~v1 and Grain-128a is a
fast correlation attack~\cite{DBLP:conf/crypto/TodoIMAZ18}. However, the attack requires a huge number of keystream bits to be generated from a single pair of key and 
initialisation vector (IV). In particular, for Grain~v1 and Grain-128a, the fast correlation attack in~\cite{DBLP:conf/crypto/TodoIMAZ18} requires $2^{75.1}$ and $2^{113.8}$ 
bits of keystream respectively to be generated from a single key and IV pair. 

The two versions of Grain, namely Grain~v1 targeted at the 80-bit security level, and Grain-128a (and its variant Grain-128AEADv2) targeted at the 128-bit security level,
have held up remarkably well against various attack ideas that have been proposed over the last two decades. As mentioned above, the only known attack against full Grain~v1 and
Grain-128a is a fast correlation attack~\cite{DBLP:conf/crypto/TodoIMAZ18}, which requires a very large number of keystream bits to be generated
from a single key and IV pair. The success of certain members of the Grain family provides the motivation for performing a more comprehensive study of the family with
the goal of uncovering other members of the family which are also secure. In this paper we perform such a comprehensive study. An overview of our contributions is
given below.

\paragraph{Definition of an abstract family.}
There are several aspects that distinguish the Grain family of stream ciphers. These are listed below.
\begin{enumerate}
	\item A combination of an NFSR and an LFSR is used, where the least significant bit of the LFSR is fed into the most significant bit of the NFSR. 
	\item The LFSR connection polynomial is a low weight binary primitive polynomial. 
	\item The feedback function of the NFSR consists of a nonlinear function $g$ to which certain other bits of the NFSR are added.
	\item The function providing the output keystream bit consists of a nonlinear function $h$ which is applied to certain bits of the NFSR and the LFSR, and to
		which some other bits of the NFSR and the LFSR are added.
	\item There is a mechanism to load the secret key and the IV into the NFSR and the LFSR.
	\item There is an initialisation procedure during which no keystream bit is produced, the goal being to obtain a good mix of the key and the IV.
	\item There is a next state function during initialisation which is different from the next state function during keystream generation.
\end{enumerate}
Our first contribution is to capture the above aspects in an abstract setting. We define appropriate parameters and mappings which allow us to describe the operation
of the Grain family of stream ciphers at an abstract level. Concrete instantiations can then be obtained by suitably defining the values of the parameters.

\paragraph{New designs of constituent Boolean functions.}
One of the main features of the Grain family consists of the two nonlinear Boolean functions $g$ and $h$. To a large extent, the security of any particular
member of the Grain family depends on the cryptographic properties of the two specific functions defined for that particular member. At the same time, for the stream cipher to 
be suitable for efficient hardware implementation, the gate counts
of $g$ and $h$ should be small. This tension between good cryptographic properties and low gate counts is typically not addressed in the literature on 
construction of Boolean functions with good cryptographic properties. In most works in the latter direction, the typical question that is studied is that of obtaining
Boolean functions which achieve optimal (or, close-to-optimal, or at least good in some sense) trade-offs between various cryptographic properties. 
The issue of whether the constructed functions have sufficiently low gate counts for use in stream ciphers is mostly not considered.
The gold standard of reference on cryptographically useful Boolean functions~\cite{BF-book} does not highlight low gate count Boolean functions
with good cryptographic properties. There are a few exceptions where gate counts of Boolean functions have received serious attention. 
Examples are the work~\cite{DBLP:conf/eurocrypt/MeauxJSC16} on design of stream cipher for use in fully homomorphic encryption, and
the work~\cite{cryptoeprint:2025/160} which explicitly considers gate counts of Boolean functions for use in the classical nonlinear filter model of stream cipher. 

The nonlinear functions $g$ and $h$ proposed for various members of the Grain family are standalone functions. The papers which introduced Grain~v1 and Grain-128a mentioned
the nonlinearity, resiliency, and algebraic degree of these functions. It was, however, not clear whether these functions are obtained from a family of functions possessing
similar properties. Perhaps more importantly there is no discussion on whether there are other functions with similar (or better) cryptographic properties and similar
(or lower) gate counts.

An important contribution of the present work is to systematically identify several classes of functions with good cryptographic properties and low gate counts. 
While the classes of functions are of some theoretical interest in their own rights, here we focus on their application to the Grain family of stream ciphers.
Grain~v1 uses a 5-variable function as $h$ and a 10-variable function as $g$. We identify a 5-variable function $h$ and a 10-variable function $g$ different from those used 
in Grain~v1 having better cryptographic properties \textit{and} lower gate counts than the corresponding functions used in Grain~v1. We also identify a 7-variable
function with better cryptographic properties than the function $h$ of Grain~v1 \textit{and} a lower gate count. 
Grain-128a uses a 9-variable function as $h$ and a 24-variable function as $g$. We identify a 10-variable function $h$ and a 24-variable function $g$ different
from the functions used in Grain-128a which provide better cryptographic properties compared to the functions $h$ and $g$ of Grain-128a, and require only a few extra
gates for implementation. From our family of Boolean functions we are able to select appropriate functions to instantiate new proposals which are targeted at the
192-bit and the 256-bit security levels.

\paragraph{New state update function during initialisation.}
An important aspect of the design of the Grain family is the next state function used during initialisation. All proposed members of the Grain family use the same next state
function during initialistion. The principle is that along with the next state function for keystream generation, the keystream bit that is generated from the
state is not provided as output, and instead is fed back to both the NFSR and the LFSR. The idea behind this strategy is to obtain a good mix of the key and the IV
in both the NFSR and the LFSR portions of the state. We examine this principle carefully. There are two nonlinear functions involved in the Grain family, namely 
$g$ and $h$. Since the function $g$ is the core nonlinear feedback function of the NFSR, its effect is fed back to the NFSR both during state initialisation as well
as during keystream generation. The function $h$ is the core nonlinear function of the keystream generation, and its effect is fed back to the NFSR and the LFSR during state
initialisation. This strategy causes an asymmetry during state initialisation. While the effect of both $g$ and $h$ are fed back to the NFSR, only the effect of $h$ is
fed back to the LFSR. The asymmetry is particularly glaring, since it is the function $g$ which is bigger and more ``complex'' than the function $h$. Having identified
this undesirable imbalance in state initialisation, we put forward a new proposal for state updation during initialisation. Our proposal is to feed back the XOR of
$g$ and $h$ to both the LFSR and the NFSR. This, however, cannot be directly done since otherwise the state updation operation may not be invertible and hence
lose entropy. We provide the exact 
definition of the new state updation function during initialisation and prove that it is indeed invertible. The net effect of the new state updation function
during initialisation is to obtain a more uniform mix of key and IV in both the NFSR and the LFSR which is more ``complex'' than the state updation function during
initialisation used in the existing members of the Grain family.

\paragraph{A systematic method of choosing tap positions.}
The choice of tap positions for obtaining the inputs to the various functions is an important design decision. The tap positions for the feedback function
of the LFSR are determined by the choice of the connection polynomial. On the other hand, for the nonlinear feedback function of the NFSR and also the output
function, presently there is no guideline on how to choose the tap positions. We made a deep and general analysis of the strategy for obtaining linear approximations
that was used in~\cite{DBLP:conf/crypto/TodoIMAZ18}. Based on this analysis, we identified a combinatorial condition on some of the tap positions which allows
certain correlations appearing in the overall correlation of the linear approximation to be provably low. Based on this combinatorial condition, we put forward
a general systematic method of choosing all the tap positions. Such a general method is advantageous since it provides a guideline for choosing tap positions
for new instantiations of the Grain family of stream ciphers.

\paragraph{Allowing different sizes for the LFSR and the NFSR.}
In all previous proposals of the Grain family the NFSR and the LFSR are of the same size. In our abstract formulation of the Grain family, we allow the NFSR and the LFSR
sizes to be different. We build upon this idea and put forward new concrete proposals at the 128-bit, 192-bit, and the 256-bit security levels where the LFSR is smaller
than the NFSR. This leads to the situation where the size of the state of the stream cipher is less than twice the size of the secret key. It is known that this condition
opens up the construction to a time/memory trade-off attack. We argue that the memory requirement of such attacks on 128-bit and higher security levels make them physically 
infeasible. In particular, we note that for the 192-bit and the 256-bit levels, the attacks would require memory which is more than the number of atoms in the entire 
world (see Appendix~\ref{app-bnd-keystream-mem}). 
From a practical engineering point of view, one may perhaps ignore attacks which are physically impossible to mount. 

\paragraph{New concrete proposals.}
We put forward seven new concrete proposals, one targeted at the 80-bit security level, and two each targeted at the 128-bit, 192-bit and the 256-bit security levels.
For the proposals at the 128-bit and higher security levels, one proposal has the NFSR and the LFSR to be of the same length, while the other proposal has the
LFSR to be of smaller length than the NFSR. We provide detailed discussions on the security of the new proposals against the various classes of attacks on the Grain family
that have been proposed till date. Compared to Grain~v1, our proposal at the 80-bit security level offers better cryptographic properties for the core 
nonlinear functions $g$ and $h$ \textit{and} lower gate counts. Compared to Grain-128a, our first proposal (one with equal length NFSR and LFSR) provides
better cryptographic properties for $g$ and $h$ and requires a few extra gates, while our second proposal (one with the LFSR smaller than the NFSR) provides better cryptographic
properties for $g$ and $h$ and has a smaller gate count (due to the smaller size of the state). For the 192-bit and the 256-bit security levels, we
know of no other stream cipher which provides smaller gate counts.

\paragraph{Scalability.}
An important advantage of our approach is scalability. To obtain a stream cipher targeted at any desired security level, one may start with our abstract description
and then appropriately obtain the nonlinear functions $g$ and $h$ from the family of functions that we have introduced, and also obtain the tap positions using
the general method that we describe in this paper. This is to be contrasted with Grain~v1 and Grain-128a appearing as kind of ad-hoc constructions at the 80-bit and 
the 128-bit security levels respectively. An imporant consequence of scalability is that it is possible to quickly scale up the parameters to improve resistance 
against any future attacks. Such scaling up may consist of either increasing the size of the NFSR, the size of the LFSR, or using larger functions for $g$ or $h$.


\paragraph{Structure of the paper.} The background and preliminaries are given in Section~\ref{sec-prelim}. In Section~\ref{sec-grain-family} we provide the 
abstract definition of the Grain family of stream ciphers. New classes of Boolean functions required to obtain new instantiations of the Grain family are
presented in Section~\ref{sec-cons}. New designs for the other components of the Grain family are given in Section~\ref{sec-design-comp}. In Section~\ref{sec-concrete}
we present the seven new concrete proposals. Analysis of known attacks and their applicability to the new proposals are described in Section~\ref{sec-attacks}, while
a possible new attack is described in Section~\ref{sec-st-guess}. Finally, we conclude the paper in Section~\ref{sec-conclu}.

\section{Background and Preliminary Results \label{sec-prelim}}
The cardinality of a finite set $S$ will be denoted by $\#S$. 
Given a subset $S$ of the set of integers and an integer $t$, we define $t+S=\{t+s:s\in S\}$. For two subsets $S$ and $T$ of the set of integers, we define 
$S+T=\{s+t:s\in S, t\in T\}=\cup_{t\in T}(t+S)=\cup_{s\in S}(s+T)$. 
For a real number $x$, we define $\sym{sgn}(x)=1$, if $x\geq 0$ and $\sym{sgn}(x)=1$, if $x<0$.

By $\mathbb{F}_2$ we denote the finite field of two elements; by $\mathbb{F}_2^n$, where $n$ is a positive integer, we denote the vector space of dimension $n$ over $\mathbb{F}_2$. 
The addition operator over both $\mathbb{F}_2$ and $\mathbb{F}_2^n$ will be denoted by $\oplus$. The product (which is also the logical AND) of $x,y\in\mathbb{F}_2$
will simply be written as $xy$.
A bit vector of dimension $n$, i.e. an element of $\mathbb{F}_2^n$ will also be considered to be an $n$-bit binary string.

Let $n$ be a positive integer. The support of $\mathbf{x}=(x_1,\ldots,x_n)\in \mathbb{F}_2^n$ is $\sym{supp}(\mathbf{x})=\{1\leq i\leq n: x_i=1\}$,
and the weight of $\mathbf{x}$ is $\sym{wt}(\mathbf{x})=\#\sym{supp}(\mathbf{x})$. 
By $\mathbf{0}_n$ and $\mathbf{1}_n$ we will denote the all-zero and all-one strings of length $n$ respectively. 
For $\mathbf{x},\mathbf{y}\in\mathbb{F}_2^n$, with $\mathbf{x}=(x_1,\ldots,x_n)$ and $\mathbf{y}=(y_1,\ldots,y_n)$ the distance between $\mathbf{x}$ and $\mathbf{y}$
is $d(\mathbf{x},\mathbf{y})=\#\{i: x_i\neq y_i\}$; the inner product of $\mathbf{x}$ and $\mathbf{y}$ is 
$\langle \mathbf{x},\mathbf{y}\rangle = x_1y_1 \oplus \cdots \oplus x_ny_n$. 

\paragraph{Boolean function.}
An $n$-variable Boolean function $f$ is a map $f:\mathbb{F}_2^n\rightarrow \mathbb{F}_2$. 
The weight of $f$ is $\sym{wt}(f)=\#\{\mathbf{x}\in\mathbb{F}_2^n:f(\mathbf{x})=1\}$; $f$ is said to be \textit{balanced} if $\sym{wt}(f)=2^{n-1}$. 
The function $f$ is said to be non-degenerate on the $i$-th variable if there are $\bm{\alpha},\bm{\beta}\in\mathbb{F}_2^n$ such that $\bm{\alpha}$ and
$\bm{\beta}$ differ only in the $i$-th position, and $f(\bm{\alpha})\neq f(\bm{\beta})$; if there is no such pair $(\bm{\alpha},\bm{\beta})$, then $f$ is said
to be degenerate on the $i$-th variable.

The \textit{algebraic normal form (ANF) representation} of an $n$-variable Boolean function $f$ is the representation of $f$ as an element of the 
polynomial ring $\mathbb{F}_2[X_1,\ldots,X_n]/(X_1^2\oplus X_1,\ldots,X_n^2\oplus X_n)$ in the following manner:
$f(X_1,\ldots,X_n) = \bigoplus_{\bm{\alpha}\in\mathbb{F}_2^n} a_{\bm{\alpha}} \mathbf{X}^{\bm{\alpha}}$, where 
$\mathbf{X}=(X_1,\ldots,X_n)$; for $\bm{\alpha}=(\alpha_1,\ldots,\alpha_n)\in \mathbb{F}_2^n$, $\mathbf{X}^{\bm{\alpha}}$ denotes the monomial
$X_1^{\alpha_1}\cdots X_n^{\alpha_n}$; and $a_{\bm{\alpha}} \in \mathbb{F}_2$. 
The (algebraic) degree of $f$ is $\sym{deg}(f)=\max\{\sym{wt}(\bm{\alpha}):a_{\bm{\alpha}}=1\}$; we adopt the convention that the zero function has degree 0.
The degree (or sometimes also called the length) of the monomial $\mathbf{X}^{\bm{\alpha}}$ is $\sym{wt}(\bm{\alpha})$.

%

Functions of degree at most 1 are said to be affine functions. Affine functions with $a_{\mathbf{0}_n}=0$ are said to be linear functions. 
Each $\bm{\alpha}=(\alpha_1,\ldots,\alpha_n)\in\mathbb{F}_2^n$, defines the linear function 
$\langle\bm{\alpha},\mathbf{X} \rangle = \langle \bm{\alpha},(X_1,\ldots,X_n)\rangle=\alpha_1X_1\oplus\cdots\oplus \alpha_nX_n$. If $\sym{wt}(\bm{\alpha})=w$, then the function
$\langle \bm{\alpha},\mathbf{X} \rangle$ is non-degenerate on exactly $w$ of the $n$ variables. 

\paragraph{Nonlinearity and Walsh transform.}
The distance between two $n$-variable functions $f$ and $g$ is $d(f,g)=\#\{\mathbf{x}\in\mathbb{F}_2^n:f(\mathbf{x})\neq g(\mathbf{x})\}=\sym{wt}(f\oplus g)$.
The \textit{nonlinearity} of an $n$-variable function $f$ is defined to be 
$\sym{nl}(f) = \min_{\bm{\alpha}\in\mathbb{F}_2^n} \{d(f,\langle \bm{\alpha},\mathbf{X} \rangle), d(f,1\oplus \langle \bm{\alpha},\mathbf{X} \rangle)\}$, 
i.e. the nonlinearity of $f$ is the minimum of the distances of $f$ to all the affine functions. 
We define the \textit{linear bias} of an $n$-variable function $f$ to be $\sym{LB}(f)=2(1/2-\sym{nl}(f)/2^n)$. From a cryptographic point of view, 
the linear bias, rather than the nonlinearity, is of importance, since it is the linear bias which is used to quantify the resistance to (fast) correlation attacks.

The Walsh transform of an $n$-variable function $f$ is the map $W_f:\mathbb{F}_2^n\rightarrow \mathbb{Z}$, where for $\bm{\alpha}\in\mathbb{F}_2^n$,
$W_f(\bm{\alpha}) = \sum_{\mathbf{x}\in\mathbb{F}_2^n} (-1)^{f(\mathbf{x}) \oplus \langle \bm{\alpha}, \mathbf{x} \rangle}$. 
From the definition it follows that $W_f(\bm{\alpha})=2^n-2d(f,\langle \bm{\alpha},\mathbf{X} \rangle)$ which shows 
$\sym{nl}(f) = 2^{n-1} - \frac{1}{2}\max_{\bm{\alpha} \in \mathbb{F}_2^n} \lvert W_f(\bm{\alpha})\rvert$. 

\paragraph{Bias and correlation.}
The bias of a random bit $b$, denoted as $\sym{bias}(b)$, is defined to be $\sym{bias}(b)=1-2\Pr[b=1]\in [-1,1]$. The correlation of an $n$-variable Boolean function $f$
at a point $\bm{\alpha}\in\mathbb{F}_2^n$ is defined to be the correlation between $f(\mathbf{x})$ and $\langle \bm{\alpha},\mathbf{X}\rangle$ and is given as
$\sym{corr}_f(\bm{\alpha})=W_f(\bm{\alpha})/2^n=1-2\Pr[f(X)\neq \langle \bm{\alpha},\mathbf{X} \rangle]$, where the
probability is over the uniform random choice of $X$ in $\mathbb{F}_2^n$; in particular, $\sym{corr}_f(\bm{\alpha})$ is equal to the bias of the random bit
$f(\mathbf{X})\oplus \langle \bm{\alpha},\mathbf{X} \rangle$.
Note that for all $\bm{\alpha}\in\mathbb{F}_2^n$, $\lvert \sym{corr}_f(\bm{\alpha})\rvert \leq \sym{LB}(f)$.

\paragraph{Bent functions.}
An $n$-variable function $f$ is said to be bent~\cite{DBLP:journals/jct/Rothaus76} if $W_f(\bm{\alpha})=\pm 2^{n/2}$ for all $\bm{\alpha}\in\mathbb{F}_2^n$.
From the definition it follows that bent functions can exist only if $n$ is even. An $n$-variable bent function has nonlinearity $2^{n-1} - 2^{n/2-1}$ 
(resp. linear bias $2^{-n/2}$), and this is the maximum possible nonlinearity (resp. least possible linear bias) that can be achieved by any $n$-variable function.

\paragraph{Resilient functions.}
Let $n$ be a positive integer and $m$ be an integer such that $0\leq m<n$. 
An $n$-variable function $f$ is said to be $m$-resilient if $W_f(\bm{\alpha})=0$ for all $\bm{\alpha}$ satisfying $\sym{wt}(\bm{\alpha})\leq m$. 
Equivalently, $f$ is $m$-resilient if and only if $d(f,\langle \bm{\alpha},\mathbf{X} \rangle)=2^{n-1}$ for $\bm{\alpha}$ satisfying $\sym{wt}(\bm{\alpha})\leq m$, i.e.
if the distance between $f$ and any linear function which is non-degenerate on at most $m$ variables is equal to $2^{n-1}$.
For notational convenience, by resiliency equal to $-1$ we mean a function that is not balanced.

\paragraph{Algebraic immunity.}
The algebraic immunity of an $n$-variable function $f$ is defined~\cite{DBLP:conf/eurocrypt/CourtoisM03,DBLP:conf/eurocrypt/MeierPC04} as follows:
$\sym{AI}(f)=\min_{g\neq 0} \{\sym{deg}(g): \mbox{ either } gf=0, \mbox{ or }  g(f\oplus 1)=0\}$.
It is known~\cite{DBLP:conf/eurocrypt/CourtoisM03} that $\sym{AI}(f)\leq \lceil n/2\rceil$. 


\paragraph{Direct sum.}
A simple way to construct a function is to add together two functions on disjoint sets of variables. The constructed function is called the direct sum of the
two smaller functions. Let $n_1$ and $n_2$ be positive integers and $g$ and $h$ be functions on $n_1$ and $n_2$ variables respectively. Define
\begin{eqnarray}\label{eqn-dsum}
	f(X_1,\ldots,X_{n_1},Y_1,\ldots,Y_{n_2}) & = & g(X_1,\ldots,X_{n_1}) \oplus h(Y_1,\ldots,Y_{n_2}).
\end{eqnarray}
Basic properties of direct sum construction are given by the following result.
\begin{proposition}[Proposition~1 of~\cite{DBLP:conf/eurocrypt/SarkarM00}]\label{prop-dsum}
	Let $f$ be constructed as in~\eqref{eqn-dsum}. Then
	\begin{compactenum}
	\item $\sym{deg}(f)=\max\{ \sym{deg}(g),\sym{deg}(h) \}$.
	\item $\sym{nl}(f) = 2^{n_1}\sym{nl}(h) + 2^{n_2}\sym{nl}(g) - 2\sym{nl}(g)\sym{nl}(h)$. Equivalently, $\sym{LB}(f)=\sym{LB}(g)\cdot \sym{LB}(h)$.
	\item $f$ is balanced if and only if at least one of $g$ or $h$ is balanced.
	\item If $g$ is $m_1$-resilient and $h$ is $m_2$-resilient, then $f$ is $m$-resilient, where $m=\max\{m_1,m_2\}$.
	\end{compactenum}
\end{proposition}

Bounds on the algebraic immunity of a function constructed as a direct sum is given by the following result.
\begin{proposition}[Lemma~3 of~\cite{DBLP:conf/eurocrypt/MeauxJSC16}]\label{prop-AI-dsum} 
	For $f$ constructed as in~\eqref{eqn-dsum}, $\max\{\sym{AI}(g),\sym{AI}(h)\} \leq \sym{AI}(f)\leq \sym{AI}(g) + \sym{AI}(h)$.
\end{proposition}

\paragraph{Maiorana-McFarland bent functions.}
The well known Maiorana-McFarland class of bent functions is defined as follows. For $k\geq 1$, let $\pi:\{0,1\}^k\rightarrow\{0,1\}^k$ be a bijection and 
$h:\{0,1\}^k\rightarrow \{0,1\}$ be a Boolean function. 
Let $\mathbf{X}=(X_1,\ldots,X_k)$ and $\mathbf{Y}=(Y_1,\ldots,Y_k)$. For $k\geq 1$, $(\pi,h)\mbox{-}\sym{MM}_{2k}$ is defined to be the following function.
\begin{eqnarray}
	(\pi,h)\mbox{-}\sym{MM}_{2k}(\mathbf{X},\mathbf{Y}) & = & \langle \pi(\mathbf{X}),\mathbf{Y}\rangle \oplus h(\mathbf{X}). \label{eqn-MM-even} 
\end{eqnarray}
Note that the degree of $(\pi,h)\mbox{-}\sym{MM}_{2k}$ is $\max(1+\max_{1\leq i\leq k}\sym{deg}(\pi_i),\sym{deg}(h))$, where $\pi_1,\ldots,\pi_k$ are the
coordinate functions of $\pi$. 
The following was proved in~\cite{cryptoeprint:2025/160}.
\begin{theorem}[Theorem~2 of~\cite{cryptoeprint:2025/160}]\label{thm-MMMaj-AI}
	Let $k\geq 2$, $n=2k$, and $\pi:\{0,1\}^k\rightarrow \{0,1\}^k$ be an affine map, i.e. each of the coordinate functions of $\pi$ is an affine function of 
	the input variables. Then $\sym{AI}((\pi,h)\mbox{-}\sym{MM}_{2k}) \geq \sym{AI}(h)$.
\end{theorem}

\paragraph{Gate count.} From an implementation point of view, it is of interest to obtain functions which are efficient to implement. A measure of
implementation efficiency is the number of gates required to implement an $n$-variable function $f$. 
For the functions that we define, we will provide counts of NOT, XOR and AND gates. By $i$[N]+$j$[X]+$k$[A] gates we will mean $i$ NOT, $j$ XOR, and $k$ AND gates.

\section{Grain Family of Stream Ciphers: an Abstraction \label{sec-grain-family} }
An elegant approach to the construction of hardware friendly stream ciphers was adopted in the design of the Grain family of stream ciphers.
In this section, we distill the essential ideas of these constructions to provide an abstract description of the Grain family of stream ciphers. 
The parameters are the following.
\begin{compactitem}
\item $\kappa,v,\kappa_1\geq \kappa,\kappa_2,a,n_0,n_1,p_0,p_1,q_0,q_1$ are non-negative integers, $n=n_0+n_1$, $p=p_0+p_1$, and $q=q_0+q_1$; 
\item $\tau(x)=x^{\kappa_2}\oplus c_{\kappa_2-1}x^{\kappa_2-1}\oplus \cdots\oplus c_{1}x\oplus 1$ is a primitive polynomial of degree $\kappa_2$ over $\mathbb{F}_2$,
	$A=(0)\cup (\kappa_2-i:c_i=1)$, and $a=\#A$;
\item $S_0,S_1,P_0,P_1$ are lists with entries from $\{0,\ldots,\kappa_1-1\}$, and $Q_0,Q_1$ are lists with entries from $\{0,\ldots,\kappa_2-1\}$ such that
	no entry is repeated in any of the lists, and $\#S_0=n_0$, $\#S_1=n_1$, $\#P_0=p_0$, $\#P_1=p_1$, $\#Q_0=q_0$ and $\#Q_1=q_1$. 
\item $g$ and $h$ are $n_0$-variable and $(p_0+q_0)$-variable Boolean functions respectively;
\item $\psi$ is a bit permutation of $p_0+q_0$ bits;
\item $\delta$ is a positive integer. 
\end{compactitem}

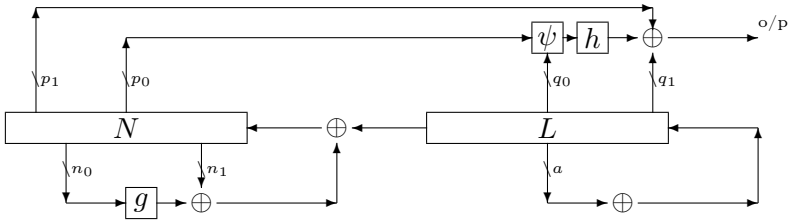
\begin{figure}[htb]
\setlength{\unitlength}{4mm}
\begin{picture}(30,10)


	\put(4,4){\framebox(8,1){$N$}}
	\put(18,4){\framebox(8,1){$L$}}
	\put(14.75,4.25){\makebox(0.5,0.5){$\oplus$}}
	\put(18,4.5){\vector(-1,0){2.5}}
	\put(14.5,4.5){\vector(-1,0){2.5}}

	\put(22,4){\vector(0,-1){2}} \put(22.05,3){\makebox(0.25,0.25){{\tiny $\backslash a$}}}
	\put(22,2){\vector(1,0){2}}
	\put(24.25,1.75){\makebox(0.5,0.5){$\oplus$}}
	\put(25,2){\vector(1,0){4}}
	\put(29,2){\vector(0,1){2.5}}
	\put(29,4.5){\vector(-1,0){3}}

	\put(6,4){\vector(0,-1){2}} \put(6.25,3){\makebox(0.25,0.25){{\tiny $\backslash n_0$}}}
	\put(6,2){\vector(1,0){2}}
	\put(8,1.5){\framebox(1,1){$g$}}
	\put(9,2){\vector(1,0){1}}
	\put(10,1.5){\makebox(1,1){$\oplus$}}
	\put(11,2){\vector(1,0){4}}
	\put(15,2){\vector(0,1){2}}
	\put(10.5,4){\vector(0,-1){1.5}} \put(10.75,3){\makebox(0.25,0.25){{\tiny $\backslash n_1$}}}

	\put(8,5){\vector(0,1){2.5}} \put(8.2,6){\makebox(0.25,0.25){{\tiny $\backslash p_0$}}}
	\put(8,7.5){\vector(1,0){13.5}}
	\put(22,5){\vector(0,1){2}} \put(22.2,6){\makebox(0.25,0.25){{\tiny $\backslash q_0$}}}
	\put(21.5,7){\framebox(1,1){$\psi$}} 
	\put(22.5,7.5){\vector(1,0){0.5}} \put(23,7){\framebox(1,1){$h$}} \put(24,7.5){\vector(1,0){1}}
	\put(25.5,5){\vector(0,1){2}} \put(25.7,6){\makebox(0.25,0.25){{\tiny $\backslash q_1$}}}
	\put(25,7){\makebox(1,1){$\oplus$}} \put(26,7.5){\vector(1,0){3}} \put(29,7.5){\makebox(1,1){\tiny o/p}}

	\put(5,5){\vector(0,1){3.5}} \put(5.2,6){\makebox(0.25,0.25){{\tiny $\backslash p_1$}}}
	\put(5,8.5){\vector(1,0){20.5}}
	\put(25.5,8.5){\vector(0,-1){0.75}}

\end{picture}
	\caption{Schematic diagram of keystream generation by the Grain family of stream cipher. \label{fig-grain} }
\end{figure}

\paragraph{Key, IV and state.} A $\kappa$-bit key $K=(K_0,\ldots,K_{\kappa-1})$ and a $v$-bit initialisation vector $\sym{IV}=(V_0,\ldots,V_{v-1})$ are used. 
The state of the stream cipher consists of a $\kappa_1$-bit register $N=(\eta_0,\ldots,\eta_{\kappa_1-1})$, and 
a $\kappa_2$-bit register $L=(\lambda_0,\ldots,\lambda_{\kappa_2-1})$. The register $N$ implements an NFSR, while the register $L$ implements an LFSR. 
During initialisation, the $\kappa$-bit key will be loaded into $N$. The condition $\kappa_1\geq \kappa$ ensures that this can be done.

\paragraph{Notation.} We will use the following notation:
given a list $T=(i_1,\ldots,i_t)$ of elements from $\{0,\ldots,\ell-1\}$ and $\mathbf{b}=(b_0,\ldots,b_{\ell-1})\in\mathbb{F}_2^\ell$, we define 
$\sym{proj}(T,\mathbf{b}) = (b_{i_1},\ldots,b_{i_t}).$ Note that $i_1,\ldots,i_t$ are not necessarily in increasing (or decreasing) order. Given a list of
bits $(b_1,\ldots,b_l)$, by $\sym{xor}(b_1,\ldots,b_l)$ we will denote the sum $b_1\oplus\cdots \oplus b_l$.

\paragraph{Maps.} The maps required to define the abstract Grain family of stream ciphers are the following. 
\begin{eqnarray*}
	\begin{array}{l}
	\sym{load}(K,\sym{IV}): (K,\sym{IV}) \mapsto (N,L); \\
	\sym{NLB}(L): L=(\lambda_0,\ldots,\lambda_{\kappa_2-1}) \mapsto \sym{xor}(\sym{proj}(A,L)) 
		= \lambda_0\oplus c_{\kappa_2-1}\lambda_1 \oplus \cdots \oplus c_1\lambda_{\kappa_2-1}; \nonumber \\
	\sym{NLS}(L): L=(\lambda_0,\ldots,\lambda_{\kappa_2-1}) \mapsto (\lambda_1,\lambda_2,\ldots,\lambda_{\kappa_2-1},\sym{NLB}(L)); \\
	\sym{NNB}(N): N \mapsto \sym{xor}(\sym{proj}(S_1,N))\oplus g(\sym{proj}(S_0,N)); \\
	\sym{NNS}(N,L): (N,L)=(\eta_0,\ldots,\eta_{\kappa_1-1},\lambda_0,\ldots,\lambda_{\kappa_2-1})
		\mapsto (\eta_1,\eta_2,\ldots,\eta_{\kappa_1-1},\sym{NNB}(N)\oplus \lambda_0); \\
	\sym{NS}(N,L): (N,L) \mapsto (\sym{NNS}(N,L),\sym{NLS}(L)); \\
	\sym{OB}(N,L): (N,L) \mapsto \sym{xor}(\sym{proj}(P_1,N),\sym{proj}(Q_1,L))\oplus h(\psi(\sym{proj}(P_0,N),\sym{proj}(Q_0,L))); \\
	\sym{init}(N,L): \{0,1\}^{\kappa_1+\kappa_2} \mapsto \{0,1\}^{\kappa_1+\kappa_2}.
	\end{array}
\end{eqnarray*}

\paragraph{Explanations.} 
$\sym{NLB}$ denotes `next linear bit', $\sym{NLS}$ denotes `next linear state', $\sym{NNB}$ denotes `next nonlinear bit',
$\sym{NNS}$ denotes `next nonlinear state', $\sym{NS}$ denotes `next state', $\sym{OB}$ denotes `output bit', and $\sym{init}$ denotes `initialise'.

\paragraph{Functions $G$ and $H$.} Let $G$ and $H$ be $n$-variable and $(p+q)$-variable functions respectively such that 
\begin{eqnarray}
	\lefteqn{\sym{NNB}(N) = G(\sym{proj}(S_1,N),\sym{proj}(S_0,N))} \nonumber \\
	& = & \sym{xor}(\sym{proj}(S_1,N))\oplus g(\sym{proj}(S_0,N)) \label{eqn-G} \\
	\lefteqn{\sym{OB}(N,L) = H(\sym{proj}(Q_1,L),\sym{proj}(P_1,N),\sym{proj}(Q_0,L),\sym{proj}(P_0,N))} \nonumber \\
	& = & \sym{xor}(\sym{proj}(Q_1,L),\sym{proj}(P_1,N))+h(\psi(\sym{proj}(Q_0,L),\sym{proj}(P_0,N)))). \label{eqn-H}
\end{eqnarray}

\paragraph{Operation.} 
Based on the above maps, the operation of the Grain family of stream ciphers is shown in Algorithm~\ref{algo-grain}. A schematic diagram of the
method of output generation (Steps~\ref{step-op} and~\ref{step-ns} of Algorithm~\ref{algo-grain}) is shown in Figure~\ref{fig-grain}. 

\begin{algorithm}
\caption{Operation of the Grain family of stream ciphers.  \label{algo-grain}}
\begin{algorithmic}[1]
\Function{\sym{Grain}}{$K,\sym{IV}$} \label{B1}
  \State {\bf initialisation phase:}
  \State $(N,L)\leftarrow \sym{load}(K,\sym{IV})$ \label{step-load}
  \State $(N,L) \leftarrow \sym{init}(N,L)$ \label{step-init}
  \State {\bf keystream generation phase:}
  \For {$t\geq 0$}  \label{step-op-gen-start}
    \State output $\sym{OB}(N,L)$ \label{step-op}
	\State $(N,L)\leftarrow \sym{NS}(N,L)$ \label{step-ns}
  \EndFor \label{step-op-gen-end}
\EndFunction.
\end{algorithmic}
\end{algorithm}

At time point $t\geq 0$, 
let $(N^{(t)},L^{(t)})$ be the state, where $(N^{(0)},L^{(0)})$ is the state produced after the initialisation phase (Steps~\ref{step-load} and~\ref{step-init} of 
Algorithm~\ref{algo-grain}). We write $N^{(t)}=(\eta_t,\eta_{t+1},\ldots,\eta_{t+\kappa_1-1})$ and $L^{(t)}=(\lambda_t,\lambda_{t+1},\ldots,\lambda_{t+\kappa_2-1})$. For $t\geq 0$, 
let the keystream bits that are produced from Step~\ref{step-op} of Algorithm~\ref{algo-grain} be denoted as $z_0,z_1,z_2,\ldots$. 
Then the loop in Steps~\ref{step-op-gen-start} to~\ref{step-op-gen-end} can be written as follows:
\begin{tabbing}
\ \ \ \ \=\ \ \ \ \=\ \ \ \ \=\ \ \ \ \kill
\> for $t\geq 0$ do \\
\>\> $z_t\leftarrow \sym{OB}(N^{(t)},L^{(t)})$; \\
\>\> $(N^{(t+1)},L^{(t+1)})\leftarrow \sym{NS}(N^{(t)},L^{(t)})$; \\
\> end for.
\end{tabbing}
From the description of Algorithm~\ref{algo-grain} and the definitions of the various maps, we have
\begin{eqnarray}
	\lefteqn{z_t=\sym{OB}(N^{(t)},L^{(t)})=H(\sym{proj}(Q_1,L^{(t)}),\sym{proj}(P_1,N^{(t)}),\sym{proj}(Q_0,L^{(t)}),\sym{proj}(P_0,N^{(t)}))} \nonumber \\
	& = & \sym{xor}(\sym{proj}(Q_1,L^{(t)})) \oplus \sym{xor}(\sym{proj}(P_1,N^{(t)}))\oplus h(\psi(\sym{proj}(Q_0,L^{(t)}),\sym{proj}(P_0,N^{(t)}))), \label{eqn-op-bit} \\
	\lambda_{t+\kappa_2}
	& = & \sym{NLB}(L^{(t)}) = \sym{xor}(\sym{proj}(A,L^{(t)}))
		= c_1\lambda_{t+\kappa_2-1}\oplus c_2\lambda_{t+\kappa_2-2}\oplus \cdots \oplus c_{t+\kappa_2-1}\lambda_{t+1}\oplus \lambda_t, \label{eqn-nlb} \\
	\eta_{t+\kappa_1}
	& = & \lambda_t \oplus \sym{NNB}(N^{(t)}) = G(\sym{proj}(S_1,N^{(t)}),\sym{proj}(S_0,N^{(t)})) \nonumber \\
	& = & \lambda_t \oplus \sym{xor}(\sym{proj}(S_1,N^{(t)}))\oplus g(\sym{proj}(S_0,N^{(t)})).  \label{eqn-nnb}
\end{eqnarray}

\paragraph{Tap positions for the feedback function of the LFSR.} These are given by the set $A$ which is determined by the polynomial $\tau(x)$. The number of tap positions is $a$.

\paragraph{Tap positions for the feedback function of the NFSR.}
The function $G$ is the nonlinear feedback function of the nonlinear feedback shift register; $n$ denotes the number of tap positions from $N$; $S_0$ and $S_1$ provide the 
indexes of the tap positions; the positions of $N$ indexed by $S_0$ and $S_1$ (of sizes $n_0$ and $n_1$ respectively) are extracted (using $\sym{proj}(S_0,N)$ and
$\sym{proj}(S_1,N)$ respectively) and fed into $G$ to obtain the next nonlinear bit. 
The definition of $G$ uses the ``core'' nonlinear function $g$ whose input bits come from the bits of $N$ which are indexed by $S_0$, and extends $g$ by
adding the bits of $N$ which are indexed by $S_1$. 
\begin{proposition}\label{prop-NS-inv}
	If $0\in S_1$ and $0\not\in S_0$, then the map $\sym{NS}$ is invertible.
\end{proposition}
\begin{proof}
	The function $\sym{NS}$ maps the state $(N,L)$ with $N=(\eta_0,\ldots,\eta_{\kappa_1-1})$ and $L=(\lambda_0,\ldots,\lambda_{\kappa_2-1})$ to
	the state $(N^\prime,L^\prime)=(\eta_1,\ldots,\eta_{\kappa_1-1},b,\lambda_1,\ldots,\lambda_{\kappa_2-1},b^\prime)$, where
	$b=\sym{NNB}(N)\oplus \lambda_0$ and $b^\prime=\sym{NLB}(L)$. To show that $\sym{NS}$ is invertible, we argue that it is possible to recover $(N,L)$ from 
	$(N^\prime,L^\prime)$. Other than $\lambda_0$ and $\eta_0$, the other bits of $(N,L)$ are also part of $(N^\prime,L^\prime)$. 
	We show that it is possible to recover $\lambda_0$ and $\eta_0$ from $(N^\prime,L^\prime)$.

	Using $b^\prime=\sym{NLB}(L)=\lambda_0\oplus c_{\kappa_2-1}\lambda_1 \oplus \cdots \oplus c_1\lambda_{\kappa_2-1}$, we have
	$\lambda_0=b^\prime\oplus c_{\kappa_2-1}\lambda_1 \oplus \cdots \oplus c_1\lambda_{\kappa_2-1}$. So $\lambda_1,\ldots,\lambda_{\kappa_2-1}$ and $b^\prime$, 
	determine the bit $\lambda_0$, i.e. $\lambda_0$ can be obtained from $L^\prime$.

	We have $\sym{NNB}(N)=\sym{xor}(\sym{proj}(S_1,N))\oplus g(\sym{proj}(S_0,N))$. Since $0\in S_1$ and $0\not\in S_0$ we may write
	$\sym{NNB}(N)=\eta_0 \oplus \eta^\prime$, where $\eta^\prime=\sym{xor}(\sym{proj}(S_1\setminus \{0\},N)) \oplus g(\sym{proj}(S_0,N\setminus\{0\}))$, i.e.
	$\eta^\prime$ is determined by the bits $\eta_1,\ldots,\eta_{\kappa_1-1}$, i.e. by $N^\prime$. So 
	$\eta_0=\sym{NNB}(N)\oplus \eta^\prime=b\oplus \lambda_0\oplus\eta^\prime$, from which it follows that $\eta_0$ is determined by $b$, $\lambda_0$ and $\eta^\prime$. 
	So the bits $b,b^\prime,\eta_1,\ldots,\eta_{\kappa_1-1}$ and $\lambda_1,\ldots,\lambda_{\kappa_2-1}$ determine the bit $\eta_0$.
\end{proof}

\begin{remark}\label{rem-S0-S1-grain}
	The condition $0\in S_1$ and $0\not\in S_0$ in Proposition~\ref{prop-NS-inv} is fulfilled by both Grain~v1 and Grain-128a.
\end{remark}

\paragraph{Tap positions for the output bit.} The function $H$ is the output function of the stream cipher. (In~\cite{DBLP:journals/ijwmc/AgrenHJM11}, 
$H$ is called the pre-output function.)
The input bits to $H$ are obtained from both $N$ and $L$. The lists $P_0$ and $P_1$ (of sizes $p_0$ and $p_1$ respectively) 
are the indexes of the tap positions from $N$, while the lists $Q_0$ and $Q_1$ (of sizes $q_0$ and $q_1$ respectively) are the indexes of the tap positions from $L$. 
The definition of $H$ uses the ``core'' nonlinear function $h$ whose inputs come from the bits of $N$ and $L$ indexed by $P_0$ and $Q_0$ respectively, and 
is extended by adding the bits of $N$ and $L$ indexed by $P_1$ and $Q_1$ respectively. 
The permutation $\psi$ applies a bit permutation to the bits extracted from $N$ and $L$ based on $P_0$ and $Q_0$.

\paragraph{Initialisation.} The function $\sym{init}$ performs the state initialisation. Two versions of $\sym{init}$, namely
$\sym{init}_1$ and $\sym{init}_2$ are shown in Algorithm~\ref{algo-init}. Both of these functions call the function $\sym{NSI}$ (which stands for `next
state during initialisation') which is defined below. Additionally, $\sym{init}_2$ also calls the function $\sym{NS}$. 
For all previously proposed variants of Grain, $\kappa_1=\kappa_2=\kappa$ and the initialisation functions were defined with this setting of parameters.
The function $\sym{init}_1$ was defined for Grain~v1 and Grain-128a (and also for earlier variants), while the function $\sym{init}_2$ was defined
for Grain-128AEADv2 (and is a modification of another initialsation function that was proposed for an earlier version of Grain-128AEADv2, namely Grain-128AEAD).
The function $\sym{init}_2$ has three loops. The first loop of $\sym{init}_2$ updates
the state using the function $\sym{NSI}$, while the third loop updates the state using $\sym{NS}$ (i.e. the next state function used for key generation).
The second loop of $\sym{init}_2$ runs for $\kappa/2$ iterations. In each iteration of the second loop, the state is updated using $\sym{NSI}$ and two bits 
of the key $K$ are added to the most significant bits of $N$ and $L$. This procedure is called key hardening and the idea of such key hardening was
adopted from the design of the stream cipher Lizard~\cite{DBLP:journals/tosc/HamannKM17}.
Grain-128AEADv2 performs both encryption and authentication, where after the second loop of $\sym{init}_2$ is completed, the registers used for authentication
are initialised. The execution of the third loop of $\sym{init}_2$ is done simultaneously with the initialisation of the authentication registers. 
\begin{algorithm}
\caption{Descriptions of $\sym{init}_1$ and $\sym{init}_2$.  \label{algo-init}}
\begin{algorithmic}[1]
\Function{$\sym{init}_1$}{$(N,L)$} 
\For {$t\leftarrow 0$ to $2\max(\kappa_1,\kappa_2)-1$} 
    \State $(N,L)\leftarrow \sym{NSI}(N,L)$
  \EndFor 
  \State return $(N,L)$
\EndFunction.
\Function{$\sym{init}_2$}{$(N,L)$} 
  \For {$t\leftarrow 0$ to $5\max(\kappa_1,\kappa_2)/2-1$} 
    \State $(N,L)\leftarrow \sym{NSI}(N,L)$
  \EndFor 
  \For {$t\leftarrow 0$ to $\kappa/2-1$} 
    \State $(N,L)\leftarrow \sym{NSI}(N,L)$
    \State $(N,L) \leftarrow (N,L) \oplus (0^{\kappa_1-1}||K_{t}||0^{\kappa_2-1}||K_{t+\kappa/2})$
  \EndFor
  \For {$t\leftarrow 0$ to $\max(\kappa_1,\kappa_2)-1$} 
    \State $(N,L)\leftarrow \sym{NS}(N,L)$
  \EndFor
  \State return $(N,L)$
\EndFunction.
\end{algorithmic}
\end{algorithm}

The definition of $\sym{NSI}$ is as follows.
\begin{eqnarray}
\begin{array}{l}
\sym{NSI}(N,L): (N,L)=(\eta_0,\eta_1,\ldots,\eta_{\kappa_1-1},\lambda_0,\lambda_1,\ldots,\lambda_{\kappa_2-1}) \\
	\quad\quad\quad\quad\quad\quad\quad\quad\quad\quad\quad\quad\quad  
		\mapsto (\eta_1,\eta_2,\ldots,\eta_{\kappa_1-1},b, \lambda_1,\lambda_2,\ldots,\lambda_{\kappa_2-1},b^\prime), 
\end{array} \label{eqn-NSI}
\end{eqnarray}
where $b=\sym{NNB}(N)\oplus \lambda_0\oplus \sym{OB}(N,L)$ and $b^\prime= \sym{NLB}(L)\oplus \sym{OB}(N,L)$.

The difference between $\sym{NSI}$ and the state update function $\sym{NS}$ (used during keystream generation) is that for $\sym{NSI}$
the output bit $\sym{OB}(N,L)$ is added to the most significant bits of both $N$ and $L$.
\begin{proposition}\label{prop-NSI-inv}
	If $0\in S_1$, $0\not\in S_0$, $0\not\in P_0\cup P_1$ and $0\not\in Q_0\cup Q_1$, then the map $\sym{NSI}$ defined by~\eqref{eqn-NSI} is invertible.
\end{proposition}
\begin{proof}
	The function $\sym{NSI}$ maps $(N,L)$, with $N=(\eta_0,\ldots,\eta_{\kappa_1-1})$ and $L=(\lambda_0,\ldots,\lambda_{\kappa_2-1})$
	to $(N^\prime,L^\prime)=(\eta_1,\eta_2,\ldots,\eta_{\kappa_1-1},b,\lambda_1,\lambda_2,\ldots,\lambda_{\kappa_2-1},b^\prime)$, 
	where $b=\sym{NNB}(N)\oplus \lambda_0\oplus \sym{OB}(N,L)$ and $b^\prime= \sym{NLB}(L)\oplus \sym{OB}(N,L)$.

	Since $0\not\in P_0\cup P_1$ and $0\not\in Q_0\cup Q_1$, it follows that the bit $\sym{OB}(N,L)$ is uniquely determined by
	$\eta_1,\eta_2,\ldots,\eta_{\kappa_1-1}$ and $\lambda_1,\lambda_2,\ldots,\lambda_{\kappa_2-1}$. So $(N^\prime,L^\prime)$
	uniquely defines $(N^{\prime\prime},L^{\prime\prime})=(\eta_1,\eta_2,\ldots,\allowbreak \eta_{\kappa_1-1},\allowbreak \sym{NNB}(N)\allowbreak \oplus \allowbreak \lambda_0,
        \allowbreak \lambda_1,\allowbreak \lambda_2,\allowbreak \ldots,\allowbreak \lambda_{\kappa_2-1},\allowbreak \sym{NLB}(L))$ which is the
	state $\sym{NS}(N,L)$.
	Now the argument in the proof of Proposition~\ref{prop-NS-inv} shows that $\lambda_0$ and $\eta_0$ are uniquely defined from
	$(N^{\prime\prime},L^{\prime\prime})$.
\end{proof}

\begin{remark}\label{rem-S0-S1-others-grain}
	The conditions $0\in S_1$, $0\not\in S_0$, $0\not\in P_0\cup P_1$ and $0\not\in Q_0\cup Q_1$ in Proposition~\ref{prop-NSI-inv} are fulfilled by both Grain~v1 and Grain-128a.
\end{remark}

\paragraph{Rationale for $\sym{init}_2$.} Since both the functions $\sym{NS}$ and $\sym{NSI}$ are efficiently invertible, if the state of the stream cipher is obtained 
at some point of time, then the state can be retraced backwards to obtain the key. In other words, state recovery implies very efficient key recovery. Note that
state recovery itself implies that the stream cipher is broken. Nonetheless, a design goal for the Grain-128AEAD and Grain-128AEADv2 was that key recovery from state
recovery should be difficult. This led to the design of the initialisation function $\sym{init}_2$ where the secret key is re-introduced into the state after a 
certain number of iterations, and is based on an idea introduced in~\cite{DBLP:journals/tosc/HamannKM17}. The point where the key may be safely reintroduced was analysed in 
details in~\cite{Grain-128AEADv2-Eprint} and was motivated by an earlier key-from-state recovery attack~\cite{DBLP:journals/iacr/ChangT21}.

\paragraph{The parameter $\delta$.} 
The throughput can be increased by a factor of $i$ for $i\leq \delta$ by replicating the state update maps and the output bit generation map 
$i$ times.  As a result, the number of clock cycles required for initialisation reduces by a factor of $i$, and during the output generation phase, $i$ bits of the output 
are produced in each clock cycle. For simplicity of such an implementation, there should be no tap positions in the rightmost $\delta-1$ bits of both $N$ and $L$. Equivalently,
all the entries of $S_0$, $S_1$, $P_0$ and $P_1$, $Q_0$ and $Q_1$ should be at most $\kappa_1-\delta$, \textit{and} 
all the entries of $A$, $Q_0$ and $Q_1$ should be at most $\kappa_2-\delta$.

\subsection{Grain~v1 and Grain-128a \label{subsec-grain-exist}}
We show how the abstract Grain family can be instantiated to obtain Grain~v1 and Grain-128a.

\paragraph{Grain v1.} The values of the parameters are as follows. 
\begin{compactenum}
\item $\kappa=\kappa_1=\kappa_2=80$, $v=64$, $a=6$, $n_0=10$, $n_1=3$, $p_0=1$, $p_1=7$, $q_0=4$, $q_1=0$;
\item $\tau(x)=x^{80}\oplus x^{67}\oplus x^{57}\oplus x^{42}\oplus x^{29}\oplus x^{18}\oplus 1$, $A=(0,13,23,38,51,62)$; 
\item $S_0=(9,15,21,28,33,37,45,52,60,63)$, $S_1=(0,14,62)$; 
\item $P_0=(63)$, $P_1=(1,2,4,10,31,43,56)$, $Q_0=(3,25,46,64)$, $Q_1=\emptyset$;
\item The functions $G$ and $H$ are defined as follows:
	\begin{eqnarray}
		\lefteqn{G(Y_1,Y_2,Y_3,X_1,\ldots,X_{10})=Y_1\oplus Y_2\oplus Y_3\oplus g(X_1,\ldots,X_{10}), \mbox{ where}} \nonumber \\
		\lefteqn{g(X_1,\ldots,X_{10})} \nonumber \\
		& = & X_1\oplus X_3\oplus X_4\oplus X_5\oplus  X_6\oplus X_7\oplus X_8\oplus X_9 \nonumber \\
		&   & \oplus X_9X_{10} \oplus X_5X_6 \oplus X_1X_2 \oplus X_7X_8X_9 \oplus X_3X_4X_5 \nonumber \\
		&   & \oplus X_1X_4X_7X_{10} \oplus X_5X_6X_8X_9 \oplus X_2X_3X_9X_{10} \nonumber \\
		&   & \oplus X_6X_7X_8X_9X_{10} \oplus X_1X_2X_3X_4X_5 \oplus X_3X_4X_5X_6X_7X_8, \label{eqn-g-grain-v1} \\
		\lefteqn{H(X_1,\ldots,X_7,Y_1,Y_2,Y_3,Y_4,Y_5) = X_1\oplus\cdots\oplus X_7 \oplus h(Y_1,Y_2,Y_3,Y_4,Y_5), \mbox{ where}} \nonumber \\
		\lefteqn{h(Y_1,Y_2,Y_3,Y_4,Y_5)} \nonumber \\
		& = & Y_2\oplus Y_5 \oplus Y_1Y_4\oplus Y_3Y_4\oplus Y_4Y_5\oplus Y_1Y_2Y_3 \oplus Y_1Y_3Y_4 \oplus Y_1Y_3Y_5 \oplus Y_2Y_3Y_5\oplus Y_3Y_4Y_5; \label{eqn-h-grain-v1}
	\end{eqnarray}
\item $\psi$ is the identity permutation;
\item $\delta=16$;
\item The function $\sym{load}(K,\sym{IV})$ outputs $K||\sym{IV}||1^{16}$;
\item The function $\sym{init}_1$ (see Algorithm~\ref{algo-init}) is used for state initialisation.
\end{compactenum}
We rewrite~\eqref{eqn-g-grain-v1} and~\eqref{eqn-h-grain-v1} as follows.
\begin{eqnarray}
\lefteqn{g(X_1,\ldots,X_{10})} \nonumber \\
	& = & X_3(1\oplus X_4X_5(1\oplus X_1X_2 \oplus X_6X_7X_8) \oplus X_2X_9X_{10}) \nonumber \\
	&   & \oplus X_6(1\oplus X_5 \oplus X_8X_9(X_5\oplus X_7X_{10})) \oplus X_7(1\oplus X_8X_9 \oplus X_1X_4X_{10}) \nonumber \\
	&   & \oplus X_1(1\oplus X_2) \oplus X_9(1\oplus X_{10}) \oplus X_4 \oplus X_5 \oplus X_8, \label{eqn-g-grain-v1-rewrite} \\
\lefteqn{h(Y_1,Y_2,Y_3,Y_4,Y_5)} \nonumber \\
	& = & Y_2 \oplus Y_3 \oplus Y_4(Y_3(Y_1\oplus Y_5\oplus 1) \oplus Y_1\oplus Y_5) + Y_5Y_3(Y_1\oplus Y_2) \oplus Y_1Y_2Y_3. \label{eqn-h-grain-v1-rewrite}
\end{eqnarray}
Using~\eqref{eqn-g-grain-v1-rewrite} it is possible to implement $g(X_1,\ldots,X_{10})$ using 6[N]+12[X]+17[A] gates. Using~\eqref{eqn-h-grain-v1-rewrite} 
and noting that $Y_1\oplus Y_5$ occurs twice, it is possible to implement $h(Y_1,Y_2,Y_3,Y_4,Y_5)$ using 1[N]+7[X]+6[A] gates. 

\paragraph{Grain-128a.} The values of the parameters are as follows.
\begin{compactenum}
\item $\kappa=\kappa_1=\kappa_2=128$, $v=96$, $a=6$, $n_0=24$, $n_1=5$, $p_0=2$, $p_1=7$, $q_0=7$, $q_1=1$;
\item $\tau(x)=x^{128}\oplus x^{121}\oplus x^{90}\oplus x^{58}\oplus x^{47}\oplus x^{32} \oplus 1$, $A=(0,7,38,70,81,96)$; 
\item $S_0=(3,67,11,13,17,18,27,59,40,48,61,65,68,84,88,92,93,95,22,24,25,70,78,82)$, \\
	$S_1=(0,26,56,91,96)$;
\item $P_0=(12,95)$, $P_1=(2,15,36,45,64,73,89)$, $Q_0=(8,13,20,42,60,79,94)$, $Q_1=(93)$.
\item The functions $G$ and $H$ are as follows:
	\begin{eqnarray}
		\lefteqn{G(X_1,\ldots,X_5,Y_1,\ldots,Y_{14},Z_1,\ldots,Z_{10})} \nonumber \\
		& = & X_1\oplus\cdots\oplus X_5\oplus g(Y_1,\ldots,Y_{14},Z_1,\ldots,Z_{10}), \mbox{ where} \label{eqn-g0-grain-128a} \\
		\lefteqn{g(Y_1,\ldots,Y_{14},Z_1,\ldots,Z_{10})} \nonumber \\
		& = & Y_1Y_2\oplus Y_3Y_4 \oplus Y_5Y_6 \oplus Y_7Y_8\oplus Y_9Y_{10} \oplus Y_{11}Y_{12}\oplus Y_{13}Y_{14} \nonumber \\
		&   & \oplus Z_1Z_2Z_3Z_4 \oplus Z_5Z_6Z_7 \oplus Z_8Z_9Z_{10}, \label{eqn-g-grain-128a} \\
		\lefteqn{H(X_1,\ldots,X_8,Y_1,\ldots,Y_9) = X_1\oplus \cdots \oplus X_8 \oplus h(Y_1,\ldots,Y_9), \mbox{ where}} \nonumber \\
		\lefteqn{h(Y_1,\ldots,Y_9) = Y_1Y_2 \oplus Y_3Y_4 \oplus Y_5Y_6 \oplus Y_7Y_8 \oplus Y_1Y_5Y_9;} \label{eqn-h-grain-128a}
	\end{eqnarray}
\item $\psi$ is the following bit permutation of $p_0+q_0=9$ bits: 
	\begin{eqnarray*}
		\psi(b_1,\ldots,b_9) & = & (b_{8},b_1,b_{2},b_{3},b_{9},b_{4},b_{5},b_{6},b_{7});
	\end{eqnarray*}
\item $\delta=32$;
\item The function $\sym{load}(K,\sym{IV})$ outputs $K||\sym{IV}||1^{31}0$;
\item The function $\sym{init}_1$ (see Algorithm~\ref{algo-init}) is used for state initialisation.
\end{compactenum}
The function $g$ given by~\eqref{eqn-g-grain-128a} can be implemented using 9[X]+14[A] gates, while the function $h$ given by~\eqref{eqn-h-grain-128a} can be
implemented using 4[X]+6[A] gates.

\begin{remark}\label{rem-Grain-128AEADv2}
	The AEAD algorithm Grain-128AEADv2 is a variant of Grain-128a which uses the functions $\sym{init}_2$ for state initialisation. 
	The other parameters of Grain-128AEADv2 remain the same as that of Grain-128a.
\end{remark}

\paragraph{Cryptographic properties and gate counts.}
Table~\ref{tab-cmp} (in Section~\ref{sec-concrete}) provides a summary of the cryptographic properties of the functions $g$, $G$, $h$ and $H$ for Grain~v1 and Grain-128a. 
The gate counts for these functions are shown in Table~\ref{tab-gc-cmp} (also in Section~\ref{sec-concrete}). 
We note that the total gate count for Grain-128a is less than the total gate count for Grain~v1.

\section{Constructions of Boolean Functions \label{sec-cons}}
In this section, we provide constructions of certain Boolean functions which we will later use in obtaining new instantiations of the Grain family of stream ciphers.
We start by noting the following function from~\cite{DBLP:journals/iacr/Sarkar25}.
\begin{proposition}\cite{DBLP:journals/iacr/Sarkar25}\label{prop-case-n=5}
Let
\begin{eqnarray}\label{eqn-n=5-simple}
	h_5(X_1,X_2,Z_1,Z_2,Z_3) & = & Z_1 \oplus Z_2 \oplus X_1(Z_1\oplus Z_3) \oplus X_2(Z_2\oplus Z_3) \oplus X_1X_2(Z_1\oplus Z_2\oplus Z_3).
\end{eqnarray}
	The function $h_5$ is a 5-variable, 1-resilient function having degree 3, algebraic immunity 2, nonlinearity 12 (and hence linear bias equal to $2^{-2}$),
	and can be implemented using 7[X]+4[A] gates. 
\end{proposition}

\begin{remark}\label{rem-h5-grain-v1-h}
	The cryptographic properties of $h_5$ given by~\eqref{eqn-n=5-simple} are the same as those of $h$ in Grain~v1 given by~\eqref{eqn-h-grain-v1} (see Table~\ref{tab-cmp}).
	On the other hand, the function $h_5$ can be implemented using 7[X]+4[A], while the function $h$ of Grain~v1 can be implemented
	using 1[N]+7[X]+6[A] gates. So $h_5$ has a smaller gate count compared to $h$ of Grain~v1 while offering the same cryptographic properties.
\end{remark}

\begin{proposition}\label{prop-f7}
Let
\begin{eqnarray}\label{eqn-n=7}
	\lefteqn{h_7(X_1,X_2,X_3,Z_1,Z_2,Z_3,Z_4)} \nonumber \\
	& = & Z_1X_1X_2X_3 \oplus Z_1X_1X_2 \oplus Z_1X_2X_3 \oplus Z_1X_3 \oplus Z_1 \nonumber \\
	&   & \oplus Z_2X_1X_2X_3 \oplus Z_2X_1 \oplus Z_2X_2X_3 \oplus Z_2X_2 \oplus Z_2 \nonumber \\
	&   & \oplus Z_3X_1 \oplus Z_3X_2X_3 \oplus Z_4X_1X_2 \oplus Z_4X_2 \oplus Z_4X_3.
\end{eqnarray}
The function $h_7$ is a 7-variable, 1-resilient function having degree 4, algebraic immunity 3, nonlinearity 56 (and hence linear bias equal to $2^{-3}$),
	and can be implemented using 2[N]+9[X]+8[A] gates.
\end{proposition}
\begin{proof}
	The function $h_7$ given by~\eqref{eqn-n=7} can be viewed as a concatenation
	of 8 linear functions of $Z_1,Z_2,Z_3$ and $Z_4$ which are non-degenerate on 2 or 3 variables. From this it follows that $h_7$ is 1-resilient and
	has nonlinearity 56 (equivalently linear bias of $2^{-3}$). It is easy to see that the degree of $h_7$ is 4. The algebraic immunity
	is obtained through direct computation. 

	For the gate count, we rewrite $h_7$ as follows.
\begin{eqnarray*}
\lefteqn{h_7(X_1,X_2,X_3,Z_1,Z_2,Z_3,Z_4)} \nonumber \\
	& = & Z_4(X_3 \oplus X_2(1\oplus X_1)) \oplus Z_3(X_1 \oplus X_2X_3) \\
	& & \oplus Z_2(1\oplus X_1 \oplus X_2 \oplus X_2X_3(1\oplus X_1)) \oplus Z_1(1\oplus X_3 \oplus  X_2X_3(1\oplus X_1) \oplus X_1X_2)
\end{eqnarray*}
	The following sequence of operations shows the number of bit operations (which provides the gate count) required for computing $h_7$.
\begin{tabbing}
\ \ \ \ \=\ \ \ \ \=\ \ \ \ \=\ \ \ \ \kill
	\> $T_1\leftarrow 1\oplus X_1$; (1[N]) \\
	\> $T_2\leftarrow X_1X_2$; $T_3\leftarrow X_2X_3$; (2[A]) \\
	\> $T_4\leftarrow X_2T_1$; $T_5\leftarrow T_3T_1$; (2[A]) \\
	\> $Z_4(X_3\oplus T_4) \oplus Z_3(X_1\oplus T_3)\oplus Z_2(T_1\oplus X_2\oplus T_5) \oplus Z_1(1\oplus X_3\oplus T_5\oplus T_2)$. (1[N]+9[X]+4[A])
\end{tabbing}
\end{proof}

The strategy used to construct $h_7$ (and also $h_5$) is to concatenate linear functions of at least 2 variables to construct 1-resilient functions. This is a 
well known strategy (introduced in~\cite{DBLP:conf/crypto/CamionCCS91}) and the obtained functions are called Maiorana-McFarland (resilient) functions. The same 
strategy can be extended to $h_{2m+1}$, for $m\geq 1$ to
obtain 1-resilient functions with nonlinearity $2^{2m}-2^m$ and degree $m+1$. This method was proposed earlier in~\cite{DBLP:conf/isit/CarletG05}. The maximum possible degree of an
$n$-variable, 1-resilient function is $n-2$~\cite{DBLP:journals/tit/Siegenthaler84}. A previous work~\cite{DBLP:conf/eurocrypt/SarkarM00} proposed the 
construction of $(2m+1)$-variable, 1-resilient functions, having degree $2m-1$
and nonlinearity $2^{2m}-2^m$. The important aspect of low gate count was not considered in either of the previous 
works~\cite{DBLP:conf/eurocrypt/SarkarM00,DBLP:conf/isit/CarletG05}. 
The descriptions of the functions in~\cite{DBLP:conf/eurocrypt/SarkarM00,DBLP:conf/isit/CarletG05} are in general terms.
We wrote out these descriptions for $n=9$ and $n=11$. It turns out that the gate counts are quite high making such functions unsuitable for implementation in 
stream ciphers targeted for low hardware footprint. To keep the gate count low, we adopt a different and simpler approach to the construction of 1-reslient functions 
on an odd number of variables. 
Compared to the previous constructions~\cite{DBLP:conf/eurocrypt/SarkarM00,DBLP:conf/isit/CarletG05}, the simpler approach provides 1-resilient functions with the same linear 
bias, but lower degrees, and substantially smaller gate counts. Keeping practical implementation in mind, we find the trade-off between degree and low gate count to be acceptable.
As the first step towards the construction of our desired family of 1-resilient functions, we note the following simple result.
\begin{proposition}\label{prop-h-2k}
For $k\geq 1$, we define $h_{2k}(U_1,\ldots,U_k,V_1,\ldots,V_k)$ as follows.
\begin{eqnarray}\label{eqn-h-even}
h_{2k}(U_1,\ldots,U_k,V_1,\ldots,V_k) & = & U_1V_1\oplus \cdots\oplus U_kV_k \oplus U_1U_2\cdots U_k.
\end{eqnarray}
The function $h_{2k}$ defined in~\eqref{eqn-h-even} is an MM bent function with nonlinearity $2^{2k-1}-2^{k-1}$ (equivalently linear bias $2^{-k}$),
and degree $k$, which is the highest possible degree for a $2k$-variable bent function. The function $h_{2k}$ can be implemented using $k$ XOR and $2k-2$ AND gates.
\end{proposition}
\begin{proof}
The statement on the nonlinearity and degree is easy to see. 
By writing $U_kV_k \oplus U_1U_2\cdots U_{k-1}U_k$ as $U_k(V_k\oplus U_1U_2\cdots U_{k-1})$, we note that $h_{2k}$ can be implemented using $k$ XOR and $2k-2$ AND gates. 
\end{proof}

\begin{proposition}\label{prop-h-5+2k}
For $k\geq 1$, let $h_{5+2k}$ be defined as follows.
\begin{eqnarray}\label{eqn-h-5+2k}
	\lefteqn{h_{5+2k}(X_1,X_2,Z_1,Z_2,Z_3,U_1,\ldots,U_k,V_1,\ldots,V_k)} \nonumber \\
	& = & h_5(X_1,X_2,Z_1,Z_2,Z_3) \oplus h_{2k}(U_1,\ldots,U_k,V_1,\ldots,V_k)
\end{eqnarray}
The function $h_{5+2k}$ is an $n$-variable, 1-resilient function, having degree $\max(3,k)$, and nonlinearity $2^{n-1}-2^{(n-1)/2}$ (equivalently, linear bias $2^{-(n-1)/2}$), 
where $n=5+2k$. Further, $h_{5+2k}$ can be implemented using $7+k$ XOR gates and $2k+2$ AND gates.
\end{proposition}
\begin{proof}
The function $h_{5+2k}$ defined in~\eqref{eqn-h-5+2k} is the direct sum of $h_5$ and $h_{2k}$. The statement about the resiliency and nonlinearity follows from the properties
of $h_5$ and $h_{2k}$ and Proposition~\ref{prop-dsum}. The gate count for $h_{5+2k}$ follows from the gate counts of $h_5$ and $h_{2k}$.
\end{proof}

\begin{remark}\label{rem-h-AI}
	For the MM functions in Proposition~\ref{prop-h-2k}, from Theorem~\ref{thm-MMMaj-AI} we obtain the lower bound of 1 on the algebraic immunity. 
	It is difficult to prove improved nontrivial lower bounds on the algebraic immunities of the functions in Proposition~\ref{prop-h-2k} and the
	same holds for the functions in Proposition~\ref{prop-h-5+2k}.
	Instead, we computed the algebraic immunities of some of these functions. The algebraic immunities of $h_8,h_{10},h_{12}$ and $h_{14}$ all turned out to be 3,
	and the algebraic immunities of $h_{13},h_{15},h_{17}$ and $h_{19}$ all turned out to be 4. In our new instantiations of the Grain family, we will need only
	$h_{10}$, $h_{15}$ and $h_{19}$.
\end{remark}

The functions $h$ constructed above will be used to instantiate the output functions of the new instances of the Grain family. Next we provide new constructions
of functions $g$ which will be used to instantiate the feedback function of the NFSR of the new instances of the Grain family.

The function $g$ of Grain~v1 given by~\eqref{eqn-g-grain-v1}, is a 10-variable, unbalanced function having degree 6, nonlinearity 430 (equivalently, linear bias $2^{-2.642}$) 
and algebraic immunity 4. A 10-variable bent function has nonlinearity 496 (equivalently linear bias $2^{-5}$), however, its degree can be at most 5. 
By simply adding the monomial of degree 10 to a 10-variable MM bent function, one obtains a function with nonlinearity 495 and degree 10; however,
the function has algebraic immunity only 2. Our goal is to increase both the degree and the nonlinearity while retaining algebraic immunity to be 4. 
By starting with a 10-variable MM bent function, and quite a bit of experimentation, we identified a very interesting 10-variable function which is given in the following result.
\begin{proposition}\label{prop-10var}
	Let 
	\begin{eqnarray}\label{eqn-g10}	
		\lefteqn{g_{10}(U_1,\ldots,U_5,V_1,\ldots,V_5)} \nonumber \\
		& = & U_1V_1 \oplus U_2V_2 \oplus U_3V_3\oplus U_4V_4\oplus U_5V_5 \oplus U_1U_2U_3U_4V_1V_2V_3 \oplus U_1U_2V_4V_5 \oplus U_3U_4V_5.
	\end{eqnarray}
	Then $g_{10}$ has degree 7, nonlinearity 492 (equivalently linear bias $2^{-4.678}$), and algebraic immunity 4. Further, $g_{10}$ can be implemented
	using 1[N]+6[X]+10[A] gates.

\end{proposition}
\begin{proof}
The nonlinearity and algebraic immunity were obtained using direct computation. The degree is clearly 7. To see the gate count, we write
\begin{eqnarray}\label{eqn-g10-gc}
\lefteqn{g_{10}(U_1,\ldots,U_5,V_1,\ldots,V_5)} \nonumber \\
& = & U_1V_1(1\oplus (U_2V_2)(U_3V_3)U_4) \oplus U_2V_2 \oplus U_3V_3 \oplus V_4(U_4\oplus U_1U_2V_5) \oplus V_5(U_5\oplus U_3U_4).
\end{eqnarray}
	From~\eqref{eqn-g10-gc} we observe that $g_{10}$ can be implemented using 1[N]+6[X]+10[A] gates. 
\end{proof}

The maximum possible algebraic immunity of any 10-variable function is 5~\cite{DBLP:conf/eurocrypt/CourtoisM03}. 
To obtain a 10-variable, degree 7 function, the nonlinearity of any 10-variable has to be lower than 496. Considering 10-variable, degree 7 functions
we know of no other functions in the literature which has nonlinearity 492 and algebraic immunity at least 4.
Further, the number of gates required to implement the function in~\eqref{eqn-g10} is quite small. The combination of low gate count
and the very attractive cryptographic properties make the function $g_{10}$ in~\eqref{eqn-g10} quite unique in the literature.

We also found the following interesting 10-variable function.
\begin{eqnarray} \label{eqn-f10}
	\lefteqn{f((U_1,\ldots,U_5,V_1,\ldots,V_5)} \nonumber \\
	& = & U_1V_1 \oplus U_2V_2 \oplus U_3V_3\oplus U_4V_4\oplus U_5V_5 \oplus U_1U_2U_3U_4V_1V_2V_3V_4 \oplus U_1U_2V_4V_5
\end{eqnarray}
which has nonlinearity 494 (equivalently linear bias $2^{-4.83}$), degree 8, and algebraic immunity 3. Compared to
$g_{10}$ given by~\eqref{eqn-g10}, the function $f$ given by~\eqref{eqn-f10} has higher degree and lower linear bias; however, 
it has algebraic immunity 3 which is one less than $g_{10}$ (and the function $g$ of Grain~v1). We did not wish to reduce any of the security properties of the function $g$ used
in Grain~v1, and that is why we opted to highlight $g_{10}$ instead of $f$ given by~\eqref{eqn-f10}.

\begin{remark}\label{rem-comp-g10-grain-v1-g}
	Comparing $g_{10}$ given by~\eqref{eqn-g10} to the function $g$ of Grain~v1 given by~\eqref{eqn-g-grain-v1}, we see that while both are unbalanced, 10-variable
	functions having algebraic immunity 4, $g_{10}$ has higher degree and significantly higher nonlinearity (equivalently significantly lower linear bias). 
	Further, $g_{10}$ can be implemented using a smaller number of gates compared to $g$ of Grain~v1.

	In view of the above comment and Remark~\ref{rem-h5-grain-v1-h}, if we replace the function $h$ of Grain~v1 by $h_5$ given by~\eqref{eqn-n=5-simple}, 
	and the function $g$ of Grain~v1 by $g_{10}$ given by~\eqref{eqn-g10}, then we obtain quantifiably better cryptographic properties \textit{and} at the same time a lower 
	gate count.
\end{remark}

\paragraph{Generalised triangular function.} Let $k>1$ be a positive integer and $k_1$ be the largest positive integer such that $k_1(k_1+1)/2\leq k$. Let $k_2=k_1(k_1+1)/2$. 
Define the function $E_k:\mathbb{F}_2^k\rightarrow \mathbb{F}_2$ as follows.
\begin{eqnarray}\label{eqn-T}
	E_k(X_1,\ldots,X_k)
	& = & X_1 \oplus X_2X_3 \oplus X_4X_5X_6 \oplus \cdots \oplus X_{k_2-2k_1+2}\cdots X_{k_2-k_1} \oplus X_{k_2-k_1+1}\cdots X_{k}.
\end{eqnarray}
Note that $E_k$ is the direct sum of $k_1$ monomials, where the $i$-th monomial with $1\leq i\leq k_1-1$, is of degree $i$, and the last monomial is of degree $k+k_1-k_2$. 
For the case $k=k_2$, the function $E_k$ was called a triangular function in~\cite{DBLP:conf/eurocrypt/MeauxJSC16}.
Clearly, $E_k(X_1,\ldots,X_k)$ has degree $k+k_1-k_2\geq k_1$. Further, 
$E_k(X_1,\ldots,X_k)$ can be implemented using $k_1-1$ XOR gates, and $k-k_2+k_1(k_1-1)/2$ AND gates.
The following result follows as a consequence of Lemma~18 of~\cite{DBLP:conf/eurocrypt/MeauxJSC16}.
\begin{proposition} \label{prop-T-ai}
Let $k>1$ be a positive integer and $k_1$ be the largest positive integer such that $k_1(k_1+1)/2\leq k$.
For $E_k(X_1,\ldots,X_k)$ defined using~\eqref{eqn-T}, $\sym{AI}(E_k)\geq k_1$. 
\end{proposition}

\begin{proposition}\label{prop-g-2k}
For $k\geq 6$, define $g_{2k}(U_1,\ldots,U_k,V_1,\ldots,V_k)$ as follows.
\begin{eqnarray} \label{eqn-g-2k-k>=6}
g_{2k}(U_1,\ldots,U_k,V_1,\ldots,V_k) & = & U_1V_1\oplus \cdots\oplus U_kV_k \oplus E_k(U_1,\ldots,U_k).
\end{eqnarray}
The function $g_{2k}$ given by~\eqref{eqn-g-2k-k>=6} is an MM bent function having nonlinearity $2^{2k-1}-2^{k-1}$ (equivalently linear bias
equal to $2^{-k}$), and degree $k+k_1-k_2$, where $k_1$ is the largest positive integer such that $k_1(k_1+1)/2\leq k$, and $k_2=k_1(k_1+1)/2$. 
Further, the algebraic immunity of $g_{2k}$ is at least $k_1$. 
\end{proposition}
\begin{proof}
The statement on nonlinearity and degree is easy to see. 
From Theorem~\ref{thm-MMMaj-AI}, the algebraic immunity of $g_{2k}$ is at least that of $E_k$, and hence using Proposition~\ref{prop-T-ai}, the
algebraic immunity of $g_{2k}$ is at least $k_1$.
\end{proof}

We highlight $g_{24}$, $g_{30}$ and $g_{36}$ since we will require these functions later.
\begin{eqnarray}
	g_{24}(U_1,\ldots,U_{12},V_1,\ldots,V_{12})
	& = & U_1V_1\oplus \cdots \oplus U_{12}V_{12} \nonumber \\
	&   & \oplus U_1 \oplus U_2U_3 \oplus U_4U_5U_6 \oplus U_7U_8U_9U_{10}U_{11}U_{12}, \label{eqn-g24} \\
	g_{30}(U_1,\ldots,U_{15},V_1,\ldots,V_{15})
	& = & U_1V_1\oplus \cdots \oplus U_{15}V_{15} \nonumber \\
	&   & \oplus U_1 \oplus U_2U_3 \oplus U_4U_5U_6 \oplus U_7U_8U_9U_{10} \oplus U_{11}U_{12}U_{13}U_{14}U_{15} \label{eqn-g30} \\
	g_{36}(U_1,\ldots,U_{18},V_1,\ldots,V_{18})
	& = & U_1V_1\oplus \cdots \oplus U_{18}V_{18} \nonumber \\
	&   & \oplus U_1 \oplus U_2U_3 \oplus U_4U_5U_6 \oplus U_7U_8U_9U_{10} \nonumber \\
	&   & \oplus U_{11}U_{12}U_{13}U_{14}U_{15}U_{16}U_{17}U_{18}. \label{eqn-g36} 
\end{eqnarray}
From the above, the degrees of $g_{24}$, $g_{30}$ and $g_{36}$ are 6, 5, and 8 respectively. We write 
	$g_{24}(U_1,\allowbreak \ldots,\allowbreak U_{12},\allowbreak V_1,\allowbreak \ldots,\allowbreak V_{12})$ as follows.
\begin{eqnarray*}
\lefteqn{g_{24}(U_1,\ldots,U_{12},V_1,\ldots,V_{12})} \nonumber \\
	& = & U_1(1\oplus V_1) \oplus U_2(V_2\oplus U_3) \oplus U_4(V_4\oplus U_5U_6) \oplus U_7(V_7\oplus U_8U_9U_{10}U_{11}U_{12}) \nonumber \\
	&   & \oplus U_3V_3 \oplus U_5V_5 \oplus U_6V_6 \oplus U_8V_8 \oplus U_9V_9 \oplus U_{10}V_{10} \oplus U_{11}V_{11} \oplus U_{12}V_{12}.
\end{eqnarray*}
	The above rewriting of $g_{24}$ shows that the function can be implemented using 1[N]+14[X]+17[A] gates. A similar rewriting of
	$g_{30}$ and $g_{36}$ show that these functions can be implemented using 1[N]+19[X]+21[A] gates and 1[N]+22[X]+27[A] gates respectively.


\section{Design of New Components\label{sec-design-comp}}
Section~\ref{sec-grain-family} provided an abstract description of the Grain family of stream ciphers which identified the different components of the
family. Section~\ref{sec-cons} presented new constructions of Boolean functions which can be used to instantiate the functions $g$ and $h$ of the Grain family.
In this section, we provide new designs of some of the other components of the Grain family.

\subsection{Choice of Tap Positions \label{subsec-tap-pos}} 
The tap positions are determined by the lists $S_0,S_1$, $P_0,P_1$, $Q_0,Q_1$ and $A$. 
We impose the following conditions on the tap positions.
\begin{compactenum}
\item The lists $S_0,S_1,P_0,P_1$ are pairwise disjoint, which ensures that no position of the NFSR $N$ is tapped twice. The lists $A,Q_0,Q_1$ are pairwise disjoint which ensures
	that no position of the LFSR $L$ is tapped twice.
\item $n_0$ is even.
\item $0\in S_1$.
\item $0\not \in P_0\cup P_1$ and $0\not\in Q_0\cup Q_1$ which ensures that the output bit $\sym{OB}(N,L)$ does not depend on either $\eta_0$ or $\lambda_0$.
\item All entries of $S_0,S_1,P_0,P_1$ (i.e. the tap positions of $N$) are at most $\kappa_1-\delta$, and all entries of $A,Q_0,Q_1$ (i.e. the tap positions of $L$) 
are at most $\kappa_2-\delta$. This allows speeding up initialisation and keystream generation by a factor of $i$ for any $i\leq \delta$.
\item $\#(P_1+S_0)=\#P_1 \cdot \#S_0$. This condition allows obtaining the upper bound $\sym{LB}(g)^{p_1}$ on the correlation of a direct sum of $p_1$ copies of $g$.
	See Remark~\ref{rem-dsum} (given in Section~\ref{subsec-la}) for an explanation of the usefulness of this upper bound in the context of the bias of a linear approximation of 
a sum of keystream bits.
\end{compactenum}
The restrictions on the tap positions given above fulfill the conditions for the function $\sym{NS}$ to be invertible (see Proposition~\ref{prop-NS-inv}).

\begin{remark}\label{rem-overlap}
We note that the first condition does not hold for either Grain~v1, or Grain-128a. For Grain~v1, 63 belongs to both $S_0$ and $P_3$, while for Grain-128a,
95 belongs to both $S_0$ and $P_0$.
\end{remark}

We next describe our method of choosing $S_0,S_1,P_0$ and $P_1$. First we define $P_1$ and $S_0$ as follows.
\begin{eqnarray}
	\left.\begin{array}{rcl}
	P_1 & = & (1,2,3,\ldots,p_1), \\
	S_0 & = & (1+p_1,1+2p_1,1+3p_1,\ldots,1+(n_0/2)p_1,\\
	& & \quad\quad\quad\quad 1+n_0p_1,1+(n_0-1)p_1,1+(n_0-2)p_1,\ldots,1+(n_0/2+1)p_1). 
	\end{array} \right\} \label{eqn-P1-S0}
\end{eqnarray}
Since $n_0$ is even, the definition of $S_0$ is meaningful. The rationale for the ordering of $S_0$ is the following.
In all our instantiations, $g$ is either an MM bent function, or constructed from an MM bent function on $n_0$ variables. So there are $n_0/2$ quadratic terms.
These $n_0/2$ quadratic terms are formed by multiplying the bit of $N$ indexed by the $i$-th entry of $S_0$ with the bit of $N$ indexed by the $(i+n/2)$-th entry
of $S_0$, $i=1,\ldots,n_0/2$. Reversing the second $n_0/2$ entries of $S_0$ ensures that shifts of $N$ do not result in the cancellation of any quadratic term. 

\begin{proposition}\label{prop-p1-n0}
	Suppose $P_1$ and $S_0$ are as defined in~\eqref{eqn-P1-S0}. Then $\#(P_1+S_0)=\#P_1 \cdot \#S_0$.
\end{proposition}
\begin{proof}
	Note that if we ignore ordering, then the elements of $S_0$ form the set $\{1+ip_1: 1\leq i\leq n_0\}$.
For any $j\in P_1+S_0$, we argue that it can be written uniquely as $j=k+l$ with $k\in P_1$ and $l\in S_0$. Suppose there are $k_1,k_2\in P_1$ and
	$l_1,l_2\in S_0$ such that $k_1+l_1=k_2+l_2$. Then $k_1-k_2=l_2-l_1$. Let $l_1=(p_1+1)+i_1p_1$ and $l_2=(p_1+1)+i_2p_1$ for some $0\leq i_1,i_2 \leq n_0-1$. 
	So $k_1-k_2=l_2-l_1=p_1(i_2-i_1)$ which implies that $p_1$ divides $k_1-k_2$. If $k_1\neq k_2$, then $\lvert k_1-k_2\rvert<p_1$ which contradicts that
	$p_1 \mid k_1-k_2$. 
\end{proof}
\begin{remark}\label{rem-n0-p1}
The maximum element in $S_0$ is $1+n_0p_1$. We require the condition $1+n_0p_1\leq \kappa_1-\delta$ to ensure that no tap position of $N$ is at a position greater
than $\kappa_1-\delta$. 
\end{remark}

Our procedure for choosing the tap positions consists of the following steps.
\begin{enumerate}
\item Choose the primitive polynomial $\tau(x)$ such that all entries of $A$ are at most $\kappa_2-\delta$.
\item Choose the entries of $Q_0$ and $Q_1$ randomly from the set $\{0,\ldots,\kappa_2-\delta\}\setminus A$.
\item Choose $P_1$ and $S_0$ as in~\eqref{eqn-P1-S0}.
\item Put $0$ in $S_1$. Choose the entries of $S_1\setminus \{0\}$, and the entries of $P_0$ randomly from the set $\{0,\ldots,\kappa_1-1\}\setminus (P_1\cup S_0\cup \{0\})$.
\end{enumerate}
The choices in Steps 2 and 4 are of the above type. From the set $\{0,\ldots,\ell-1\}$ choose a pair of disjoint sets of $r_1$ and $r_2$ elements avoiding 
a set of $r$ elements. The strategy that we used is the following. In an array of length $\ell-r$, set the first $r_1$ positions to 1, the next $r_2$ positions to 2,
and the rest of the elements of the array to 0. Next we perform the Fisher-Yates (see Algorithm~P in Section~3.4.2 of~\cite{Knuth-vol2}) random shuffle on the 
array a total of four times. (Fisher-Yates shuffle needs to be performed once to obtain a uniform random permutation; however, since our choice of the array index at each 
step is pseudo-random and not uniform random we perform the shuffle four times). 
After the shuffling is over, we insert the value 3 into the array at the positions of
the $r$ elements that are to be avoided. This step increases the length of the array from $\ell-r$ to $\ell$. In this array of length $\ell$, the positions
marked by 1 give the subset of size $r_1$ and the positions marked by 2 give the subset of size $r_2$.

\subsection{A New Initialisation Procedure \label{subsec-new-init}}
We define a new version of the next state function to be used during initialisation which we call $\sym{NSIG}$. The definition of $\sym{NSIG}$ is as follows.
\begin{eqnarray}
\begin{array}{l}
\sym{NSIG}(N,L): (N,L)=(\eta_0,\eta_1,\ldots,\eta_{\kappa_1-1},\lambda_0,\lambda_1,\ldots,\lambda_{\kappa_2-1}) \\
	\quad\quad\quad\quad\quad\quad\quad\quad\quad\quad\quad\quad\quad  
		\mapsto (\eta_1,\eta_2,\ldots,\eta_{\kappa_1-1}, b, \lambda_1,\lambda_2,\ldots,\lambda_{\kappa_2-1},b^\prime), 
\end{array} \label{eqn-NSIG}
\end{eqnarray}
where $b=\lambda_0\oplus \sym{NNB}(N)\oplus \sym{OB}(N,L)$ and $b^\prime=\sym{NLB}(L)\oplus b=\sym{NLB}(L)\oplus \lambda_0\oplus \sym{NNB}(N)\oplus \sym{OB}(N,L)$.

\begin{proposition}\label{prop-NSIG-inv}
	If $0\in S_1$, $0\not\in S_0$, $0\not\in P_0\cup P_1$ and $0\not\in Q_0\cup Q_1$, then the map $\sym{NSIG}$ defined by~\eqref{eqn-NSIG} is invertible.
\end{proposition}
\begin{proof}
	The function $\sym{NSIG}$ maps $(N,L)$, with $N=(\eta_0,\ldots,\eta_{\kappa_1-1})$ and $L=(\lambda_0,\ldots,\lambda_{\kappa_2-1})$
	to $(N^\prime,L^\prime)=(\eta_1,\eta_2,\ldots,\eta_{\kappa-1}, b, \lambda_1,\lambda_2,\ldots,\lambda_{\kappa-1},b^\prime)$, where
	$b=\lambda_0\oplus \sym{NNB}(N)\oplus \sym{OB}(N,L)$ and $b^\prime=\sym{NLB}(L)\oplus b$.

	Since $0\not\in P_0\cup P_1$ and $0\not\in Q_0\cup Q_1$, it follows that the bit $\sym{OB}(N,L)$ is determined by
	$\eta_1,\eta_2,\ldots,\eta_{\kappa_1-1}$ and $\lambda_1,\lambda_2,\ldots,\lambda_{\kappa_2-1}$. Further, as in Proposition~\ref{prop-NS-inv}
	since $0\in S_1$ and $0\not\in S_0$, we may write $\sym{NNB}(N,L)=\eta_0\oplus \eta^\prime$, where $\eta^\prime$ is determined
	by $\eta_1,\eta_2,\ldots,\eta_{\kappa_1-1}$. Also, $\sym{NLB}(L)=\lambda_0\oplus \lambda^\prime$, where
	$\lambda^\prime=c_{\kappa_2-1}\lambda_1 \oplus \cdots \oplus c_1\lambda_{\kappa_2-1}$ is determined by $\lambda_1,\lambda_2,\ldots,\lambda_{\kappa_2-1}$. 

	We have 
	\begin{eqnarray*}
		b^\prime 
		& = & b\oplus \sym{NLB}(L) \\
		& = & b\oplus \lambda_0\oplus \lambda^\prime \\
		& = & \sym{NNB}(N) \oplus \lambda_0 \oplus \sym{OB}(N,L) \oplus \lambda_0 \oplus \lambda^\prime \\
		& = & \eta_0\oplus \eta^\prime \oplus \sym{OB}(N,L) \oplus \lambda^\prime,
	\end{eqnarray*}
	which shows that $\eta_0=b^\prime \oplus \eta^\prime \oplus \sym{OB}(N,L) \oplus \lambda^\prime$. So the bits $b^\prime$, $\lambda_1,\lambda_2,\ldots,\lambda_{\kappa_2-1}$
	and $\eta_1,\eta_2,\ldots,\eta_{\kappa_1-1}$ determine the bit $\eta_0$. 

	Next we write $b=\lambda_0\oplus \sym{NNB}(N)\oplus\sym{OB}(N,L)=\eta_0\oplus \eta^\prime \oplus \lambda_0\oplus \sym{OB}(N,L)$ 
	which implies $\lambda_0=b\oplus \eta_0\oplus \eta^\prime\oplus \sym{OB}(N,L)$.
	So the bits $b$, $\eta_0$, $\eta^\prime$, and $\sym{OB}(N,L)$ determine the bit $\lambda_0$. From the discussion in the previous paragraphy,
	it follows that $\lambda_0$ is determined by $b$, $b^\prime$, $\lambda_1,\lambda_2,\ldots,\lambda_{\kappa_2-1}$, and $\eta_1,\eta_2,\ldots,\eta_{\kappa_1-1}$.
\end{proof}
\begin{remark}\label{rem-new-init-tap-pos}
	Our choice of tap positions described in Section~\ref{subsec-tap-pos} fulfills the conditions of Proposition~\ref{prop-NSIG-inv}.
\end{remark}

\paragraph{Difference between $\sym{NSI}$ given by~\eqref{eqn-NSI} and $\sym{NSIG}$ given by~\eqref{eqn-NSIG}.} 
Both $\sym{NSI}$ and $\sym{NSIG}$ provide nonlinear feedbacks to the most significant bits of both $N$ and $L$.
The feedbacks provided by both $\sym{NSI}$ and $\sym{NSIG}$ to the most significant bit of $N$ are the same. The difference between $\sym{NSI}$ and $\sym{NSIG}$ is
in the type of feedback provided to the most significant bit of $L$; $\sym{NSI}$ provides the feedback $\sym{NLB}(L) \oplus \sym{OB}(N,L)$, while $\sym{NSIG}$ 
provides the feedback $\lambda_0\oplus \sym{NLB}(L) \oplus \sym{OB}(N,L) \oplus \sym{NNB}(N)$. Since $\sym{NLB}(L)$ is a linear function of $L$, the nonlinear
part of the feedback to $L$ by $\sym{NSI}$ is $\sym{OB}(N,L)$, whereas the nonlinear part of the feedback to $L$ by $\sym{NSIG}$ is $\sym{OB}(N,L)\oplus \sym{NNB}(N)$.

The rationale for introducing $\sym{NSIG}$ is the following. During initialisation, the key 
$K$ is loaded to the register $N$ while the padded $\sym{IV}$ is loaded to the register $L$. Since the $\sym{IV}$ is known, the only unknown is the secret key $K$. 
A goal of the initialisation phase is that at the end of this phase all the bits of both $N$ and $L$ depend in a complex nonlinear manner on the unknown bits of the secret key
$K$. In $\sym{NSI}$, the output of the function $\sym{OB}$ is added to the most significant bit of $L$, while in $\sym{NSIG}$, along with the output of $\sym{OB}$,
the output of the function $\sym{NNB}$ is also added to the most significant bit of $L$. The function $\sym{OB}$ depends nonlinearly on $p_0$ bits of $K$ while the function 
$\sym{NNB}$ depends nonlinearly on $n_0$ bits of $K$. Further, the nonlinear function $h$ is used to determine $\sym{OB}$ while the nonlinear function $g$ is used to 
determine $\sym{NNB}(N)$.
In the proposals Grain~v1, Grain-128a, and also our new proposals given below, the condition $p_0<n_0$ holds (see Table~\ref{tab-params}). Further, in all cases, $h$ is a 
much smaller function than $g$ both in terms of the number of variables and algebraic degree (see Table~\ref{tab-cmp}). 
Due to the introduction of $\sym{NNB}(N)$ in the feedback to $L$, 
compared to the function $\sym{NSI}$, the function $\sym{NSIG}$ achieves a more complex nonlinear injection of the bits of $K$ to the register $L$. 

\begin{remark}\label{rem-lambda0-inv}
	In~\eqref{eqn-NSIG}, we have $b^\prime=\sym{NLB}(L)\oplus \lambda_0\oplus \sym{NNB}(N)\oplus \sym{OB}(N,L)$. The addition of $\lambda_0$ is important.
	If instead we had defined $b^\prime=\sym{NLB}(L)\oplus \sym{NNB}(N)\oplus \sym{OB}(N,L)$, then the corresponding function $\sym{NSIG}$ would not have
	been invertible.
\end{remark}

\paragraph{The function $\sym{initG}$.} 
We define the function 
$\sym{initG}:\{0,1\}^{\kappa_1}\times \{0,1\}^{\kappa_2}\rightarrow \{0,1\}^{\kappa_1}\times \{0,1\}^{\kappa_2}$ to be the function $\sym{init}_1$ 
given by Algorithm~\ref{algo-init} where the call to the function $\sym{NSI}$ is replaced by a call to the function $\sym{NSIG}$.

\begin{remark}\label{rem-key-from-state}
Since $\sym{NS}$ and $\sym{NSIG}$ are both efficiently invertible functions, recovery of the state at any point of time leads to recovering the key 
efficiently. This is similar to the situation for $\sym{init}_1$ used with the function $\sym{NSI}$ and has been discussed earlier in Section~\ref{sec-grain-family}.
As has been mentioned earlier, state recovery using less than $2^{\kappa}$ operations implies that the stream cipher is broken. If it is desired that
the key recovery should be difficult from state recovery, then one may use the initialsation function $\sym{init}_2$ with the call to $\sym{NSI}$ 
replaced by the call to $\sym{NSIG}$.
\end{remark}

\paragraph{The function $\sym{load}(K,\sym{IV})$.} The $\kappa$-bit key $K$ and the $v$-bit IV are used to load the state $(N,L)$ which is of length
$\kappa_1+\kappa_2$. In all our instantiations, we have $\kappa_1+\kappa_2-(\kappa+v)$ to be even. We define the function $\sym{load}(K,\sym{IV})$ to 
output $K||\sym{IV}||(01)^c$, where $c=(\kappa_1+\kappa_2-(\kappa+v))/2$. Since $\kappa_1\geq \kappa$, the key $K$ is loaded completely into $N$, the remaining
bits of $N$ and the most significant bits of $L$ are loaded with IV, and finally the left over bits of $L$ are padded with the balanced string $(10)^c$. 
In all our instantiations, we have $\kappa_1=\kappa$, so $N$ is completely filled with $K$.

An early analysis~\cite{Kucuk2006} had pointed out problems if the padding of $L$ is done by either the all-zero or the all-one string.
In response, Grain-128a~\cite{DBLP:journals/ijwmc/AgrenHJM11} had defined the padding to be the all-zero string followed by a single 1, and
had mentioned that this padding rule rendered the attacks reported in~\cite{Kucuk2006,DBLP:conf/africacrypt/CanniereKP08,DBLP:conf/acisp/LeeJSH08} inapplicable.
Our proposal of using a balanced string also has the effect of rule out these attacks. A recent work~\cite{DBLP:conf/secitc/MaimutT21} performed a more
systematic study of the padding rule in the load function. This study does not show any attack against using the balanced string for padding.

\subsection{Choice of the Permutation $\psi$ \label{subsec-perm-psi}} 
We instantiate the function $h$ (which is of $p_0+q_0$ variables) with either of $h_7$ (given by~\eqref{eqn-n=7}), 
$h_{10}$ (given by~\eqref{eqn-h-even}), $h_{15}$ or $h_{19}$ (given by~\eqref{eqn-h-5+2k}).
The inputs to these functions come from $L$ and $N$ with tap positions determined by $Q_0$ and $P_0$ respectively. The number of tap positions from $L$ is $p_0$
and the number of tap positions from $N$ is $q_0$. Recall that in the definition of the output bit function, the permutation $\psi$ is first applied 
to $(\sym{proj}(P_0,N),\sym{proj}(Q_0,L))$ and then the function $h$ is applied, i.e. the computation required is $h(\psi(\sym{proj}(P_0,N),\sym{proj}(Q_0,L)))$.
This requires specifying the permutation $\psi$. For $q_0=p_0+1$, we define the following permutation.
\begin{eqnarray}\label{eqn-psi}
	\psi_{p_0+q_0}(a_1,a_2,\ldots,a_{p_0},b_1,b_2,\ldots,b_{q_0}) & = & 
	(\underbrace{b_1,a_1,b_2,a_2,b_3},\underbrace{a_3,\, a_4,\ldots,a_{p_0}},\underbrace{b_4,b_5,\ldots,b_{q_0}}).
\end{eqnarray}
The rationale for defining the above permutation is the following. Consider $h_{15}$. From~\eqref{eqn-h-5+2k},
$$h_{15}(X_1,X_2,Z_1,Z_2,Z_3,U_1,\ldots,U_5,V_1,\ldots,V_5)=h_5(X_1,X_2,Z_1,Z_2,Z_3)\oplus h_{10}(U_1,\ldots,U_5,V_1,\ldots,V_5).$$
For $p_0=7$ and $q_0=8$, suppose $P_0=(i_1,\ldots,i_{7})$ and $Q_0=(j_1,\ldots,j_{8})$. Then $\sym{proj}(P_0,N)=(\eta_{i_1},\ldots,\eta_{i_{7}})$
and $\sym{proj}(Q_0,L)=(\lambda_{j_1},\ldots,\lambda_{j_{8}})$ and
\begin{eqnarray*}
	h_{15}(\psi_{15}(\eta_{i_1},\ldots,\eta_{i_{7}},\lambda_{j_1},\ldots,\lambda_{j_{8}}))
	& = & h_{15}(\lambda_{j_1},\eta_{i_1},\lambda_{j_2},\eta_{i_2},\lambda_{j_3},\, \eta_{i_3},\eta_{i_4},\ldots,\eta_{i_7},\lambda_{j_4},\lambda_{j_5},\ldots,\lambda_{j_8}) \\
	& = & h_5(\lambda_{j_1},\eta_{i_1},\lambda_{j_2},\eta_{i_2},\lambda_{j_3})
	\oplus h_{10}(\eta_{i_3},\eta_{i_4},\ldots,\eta_{i_7},\lambda_{j_4},\lambda_{j_5},\ldots,\lambda_{j_8}).
\end{eqnarray*}
The above ensures that the quadratic terms in the definition of $h_{10}$ given by~\eqref{eqn-h-even} 
consist of one element from $N$ and one element from $L$. Further, there is a good mix of the
$\lambda$ and $\eta$ terms in the definition of $h_5$ given by~\eqref{eqn-n=5-simple}. 
For any $q_0=3+k$ and $p_0=2+k$, we obtain a mix of the $\lambda$ and $\eta$ terms in the composition $h_{p_0+q_0}\circ \psi_{p_0+q_0}$
which is similar to the mix illustrated above for $q_0=8$ and $p_0=7$. 

\subsection{Sizes of the NFSR and the LFSR \label{subsec-LFSR-NFSR-sizes}} 
The NFSR $N$ is of length $\kappa_1$ bits and the LFSR $L$ is of length $\kappa_2$ bits. The initialisation
procedure loads the $\kappa$-bit secret key $K$ into $N$, which necessitates $\kappa_1\geq \kappa$. The total size of the state is $\kappa_1+\kappa_2$. 
Since the feedback polynomial of the LFSR is primitive, if the LFSR is loaded with a nonzero value, then the combined state of NFSR and LFSR takes distinct values for
$2^{\kappa_2}-1$ time steps. If at most $2^B$ bits of keystream are to be generated from a single pair of secret key and IV, then setting $2^{\kappa_2}>2^B$ 
(equivalently, $\kappa_2>B$) ensures that the combined state is not repeated while generating the $2^B$ bits of the keystream. 

There are several possibilities for choosing $\kappa_1\geq \kappa$ and $\kappa_2>B$ which can be divided into two broad categories depending on whether 
$\kappa_1+\kappa_2\leq 2\kappa$ or $\kappa_1+\kappa_2>2\kappa$. The various options are given below
\begin{description}
	\item{Category~1.} $\kappa_1+\kappa_2\leq 2\kappa$.
\begin{compactdesc}
	\item{Option~1.} $\kappa_1=\kappa_2=\kappa$, $\kappa_1+\kappa_2=2\kappa$.
	\item{Option~2.} $\kappa_1=\kappa$, $\kappa_2<\kappa$, $\kappa_1+\kappa_2<2\kappa$.
	\item{Option~3.} $\kappa_1>\kappa$, $\kappa_2<\kappa$, $\kappa_1+\kappa_2\leq 2\kappa$.
\end{compactdesc}
	\item{Category~2.} $\kappa_1+\kappa_2>2\kappa$.
\end{description}
Option~1 was followed by all previous proposals of the Grain family. Properly chosen proposals based on Option~2 will be smaller than proposals based on Option~1.
Option~3 provides for the NFSR to be of length greater than $\kappa$. This allows using bigger nonlinear feedback functions. So proposals for Option~3 can potentially 
provide higher security than proposals for Option~1 and Option~2 without varying the size of the hardware too much. 
Appropriate proposals for Category~2 can potentially provide higher security though the increase in the size of the state will increase the size of the stream cipher.

In the next section, we put forward four new proposals at the 80-bit, 128-bit, 192-bit and 256-bit security levels based on Option~1. We also put forward three new 
proposals based on Option~2 at the 128-bit, 192-bit and the 256-bit security levels. The methodology that we have developed in putting forward the new proposals can also 
be applied to obtain new proposals following Option~3 and Category~2. We leave these as future tasks.

\section{New Concrete Proposals of the Grain Family\label{sec-concrete}}
We put forward seven new proposals of the Grain family targeted at four different security levels, namely 80-bit, 128-bit, 192-bit and 256-bit security 
levels. The proposals are named as R or W with the appropriate security level attached. As members of the Grain family, `R' may be taken to stand for `Rice' and
`W' may be taken to stand for `Wheat'. The parameters of the new proposals
are given in Table~\ref{tab-params}. The tap positions, the feedback functions, the permutation $\psi$, and the function $\sym{load}$ are given below.

\begin{table}
{\small
\centering

	\begin{tabular}{|l|c|c||c|c|c||c||c|c||c|c||c|c||c||c|}
		\cline{2-15}
		\multicolumn{1}{c|}{ } & $\kappa$ & $v$ & $\kappa_1$ & $\kappa_2$ & $\kappa_1+\kappa_2$ 
			  				& $a$ & $n_0$ & $n_1$ & $p_0$ & $p_1$ & $q_0$ & $q_1$ & $p_0+q_0$ & $\delta$ \\ \hline
		Grain~v1 & 80 & 64 & 80 & 80 & 160 & 6 & 10 & 3 & 1 & 7 & 4 & 0 & 5 & 16 \\ \hline
		Grain-128a & 128 & 96 & 128 & 128 & 256 & 6 & 24 & 5 & 2 & 7 & 7 & 1 & 9 & 32 \\ \hline\hline
		R-80    & 80 & 64 & 80 & 80 & 160 & 6 & 10 & 3 & 3 & 6 & 4 & 1 & 7 & 16 \\ \hline\hline
		R-128   & 128 & 96 & 128 & 128 & 256 & 6 & 24 & 5 & 5 & 4 & 5 & 4 & 10 & 31 \\ \hline
		W-128   & 128 & 96 & 128 & 112 & 240 & 6 & 24 & 5 & 5 & 4 & 5 & 4 & 10 & 31 \\ \hline\hline
		R-192   & 192 & 128 & 192 & 192 & 384 & 6 & 30 & 7 & 7 & 5 & 8 & 5 & 15 & 32 \\ \hline
		W-192   & 192 & 128 & 192 & 160 & 352 & 6 & 30 & 7 & 7 & 5 & 8 & 5 & 15 & 32 \\ \hline\hline
		R-256   & 256 & 128 & 256 & 256 & 512 & 6 & 36 & 9 & 9 & 6 & 10 & 6 & 19 & 32 \\ \hline
		W-256   & 256 & 128 & 256 & 208 & 464 & 6 & 36 & 9 & 9 & 6 & 10 & 6 & 19 & 32 \\ \hline\hline
	\end{tabular}

	\caption{Values of the parameters for Grain~v1, Grain-128a, and the new instantiations. \label{tab-params} }
}
\end{table}

\paragraph{IV size.} For the 80-bit and the 128-bit security levels following Grain~v1 and Grain-128a, we propose the IV sizes to be 64 bits and 96 bits respectively.
For the 192-bit and the 256-bit security levels, we propose the IV sizes to be 128-bit and 192-bit respectively. 

\paragraph{R-80.} 
\begin{compactenum}
\item $\tau(x)=x^{80}\oplus x^{77}\oplus x^{65}\oplus x^{29}\oplus x^{19}\oplus x^{16}\oplus 1$, $A=(0,3,15,51,61,64)$;
\item $S_0=(7,13,19,25,31,\ 61,55,49,43,37)$, $S_1=(0,54,57)$; 
\item $P_0=(15, 16, 39)$, $P_1=(1,2,3,4,5,6)$, $Q_0=(5, 12, 16, 19)$, $Q_1=(11)$.
\item The functions $g$ and $h$ are as follows.
	\begin{eqnarray*}
		g(X_1,\ldots,X_{10}) & = & g_{10}(X_1,\ldots,X_{10}), \\
		h(X_1,X_2,X_3,Z_1,Z_2,Z_3,Z_4) & = & h_7(X_1,X_2,X_3,Z_1,Z_2,Z_3,Z_4),
	\end{eqnarray*}
	where $g_{10}$ is the function given in~\eqref{eqn-g10}, and $h_7$ is the function given in~\eqref{eqn-n=7}.
\item $\psi$ is the permutation $\psi_{7}$ given by~\eqref{eqn-psi}, with $7$ written as $7=p_0+q_0=3+4$
\item The function $\sym{load}(K,\sym{IV})$ outputs $K||\sym{IV}||(10)^{8}$.
\item The function $\sym{initG}$ is used for state initialisation.
\end{compactenum}

\paragraph{R-128.} 
\begin{compactenum}
\item $\tau(x)=x^{128}\oplus x^{108}\oplus x^{97}\oplus x^{54}\oplus x^{46} \oplus x^{32} \oplus 1$, $A=(0,20,31,74,82,96)$; 
\item $S_0=(5,9,13,17,21,25,29,33,37,41,45,49,\  97,93,89,85,81,77,73,69,65,61,57,53)$;
\item $S_1=(0,36, 55, 71, 91)$;
\item $P_0=(6, 31, 39, 50, 67)$, $P_1=(1,2,3,4)$, $Q_0=(1, 12, 38, 87, 97)$, $Q_1=(5, 10, 30, 85)$;
\item The functions $g$ and $h$ are as follows.
	\begin{eqnarray*}
		g(U_1,\ldots,U_{12},V_1,\ldots,V_{12}) & = & g_{24}(U_1,\ldots,U_{12},V_1,\ldots,V_{12}) \\
		h(U_1,\ldots,U_5,V_1,\ldots,V_5) & = & h_{10}(U_1,\ldots,U_5,V_1,\ldots,V_5),
	\end{eqnarray*}
	where $g_{24}$ is given by~\eqref{eqn-g24} and $h_{10}$ is given by~\eqref{eqn-h-even}.
\item $\psi$ is the identity permutation. 
\item The function $\sym{load}(K,\sym{IV})$ outputs $K||\sym{IV}||(10)^{32}$.
\item The function $\sym{initG}$ is used for state initialisation.
\end{compactenum}

\paragraph{W-128.} 
\begin{compactenum}
\item $\tau(x)=x^{112}\oplus x^{93}\oplus x^{84}\oplus x^{74}\oplus x^{43}\oplus x^{32}\oplus 1$, $A=(0,19,28,38,69,80)$; 
\item $S_0=(5,9,13,17,21,25,29,33,37,41,45,49,\  97,93,89,85,81,77,73,69,65,61,57,53)$;
\item $S_1=(0,28, 54, 67, 68)$;
\item $P_0=(11, 26, 30, 44, 76)$, $P_1=(1,2,3,4)$, $Q_0=(11, 36, 56, 73, 76)$, $Q_1=(13, 31, 39, 77)$;
\item The functions $g$ and $h$ are as follows.
	\begin{eqnarray*}
		g(U_1,\ldots,U_{12},V_1,\ldots,V_{12}) & = & g_{24}(U_1,\ldots,U_{12},V_1,\ldots,V_{12}) \\
		h(U_1,\ldots,U_5,V_1,\ldots,V_3) & = & h_{10}(U_1,\ldots,U_5,V_1,\ldots,V_3),
	\end{eqnarray*}
	where $g_{24}$ is given by~\eqref{eqn-g24} and $h_{10}$ is given by~\eqref{eqn-h-even}.
\item $\psi$ is the identity permutation. 
\item The function $\sym{load}(K,\sym{IV})$ outputs $K||\sym{IV}||(10)^{8}$.
\item The function $\sym{initG}$ is used for state initialisation.
\end{compactenum}

\paragraph{R-192.} 
\begin{compactenum}
\item $\tau(x)=x^{192}\oplus x^{131} \oplus x^{123}\oplus x^{118} \oplus x^{79} \oplus x^{32}\oplus 1$, $A=(0,61,69,74,113,160)$;
\item $S_0=(6, 11, 16, 21, 26, 31, 36, 41, 46, 51, 56, 61, 66, 71, 76,$ \\
	$151, 146, 141, 136, 131, 126, 121, 116, 111, 106, 101, 96, 91, 86, 81)$; \\
	$S_1=(0,22, 68, 75, 82, 89, 129)$;
\item $P_0=(35, 69, 83, 88, 98, 104, 150)$, $P_1=(1,2,3,4,5)$, \\
	$Q_0=(1, 26, 57, 77, 83, 103, 116, 127)$, $Q_1=(60, 75, 101, 122, 123)$;
\item The functions $g$ and $h$ are as follows.
	\begin{eqnarray*}
		g(U_1,\ldots,U_{15},V_1,\ldots,V_{15}) & = & g_{30}(U_1,\ldots,U_{15},V_1,\ldots,V_{15}) \\
		h(X_1,X_2,Z_1,Z_2,Z_3,U_1,\ldots,U_5,V_1,\ldots,V_5) & = & h_{15}(X_1,X_2,Z_1,Z_2,Z_3,U_1,\ldots,U_5,V_1,\ldots,V_5),
	\end{eqnarray*}
	where $g_{30}$ is given by~\eqref{eqn-g30} and $h_{15}$ is given by~\eqref{eqn-h-5+2k}.
\item $\psi$ is the permutation $\psi_{15}$ given by~\eqref{eqn-psi}, with 15 written as $15=p_0+q_0=7+8$.
\item The function $\sym{load}(K,\sym{IV})$ outputs $K||\sym{IV}||(10)^{32}$.
\item The function $\sym{initG}$ is used for state initialisation.
\end{compactenum}

\paragraph{W-192.} 
\begin{compactenum}
\item $\tau(x)=x^{160}\oplus x^{142} \oplus x^{76}\oplus x^{57} \oplus x^{44} \oplus x^{32}\oplus 1$, $A=(0,18,84,103,116,128)$;
\item $S_0=(6, 11, 16, 21, 26, 31, 36, 41, 46, 51, 56, 61, 66, 71, 76,$ \\
	$151, 146, 141, 136, 131, 126, 121, 116, 111, 106, 101, 96, 91, 86, 81)$; \\
	$S_1=(0,43, 53, 72, 75, 123, 140)$;
\item $P_0=(30, 54, 58, 80, 112, 156, 160)$, $P_1=(1,2,3,4,5)$, \\
	$Q_0=(10, 43, 51, 91, 96, 110, 111, 127)$, $Q_1=(8, 26, 108, 113, 115)$;
\item The functions $g$ and $h$ are as follows.
	\begin{eqnarray*}
		g(U_1,\ldots,U_{15},V_1,\ldots,V_{15}) & = & g_{30}(U_1,\ldots,U_{15},V_1,\ldots,V_{15}) \\
		h(X_1,X_2,Z_1,Z_2,Z_3,U_1,\ldots,U_5,V_1,\ldots,V_5) & = & h_{15}(X_1,X_2,Z_1,Z_2,Z_3,U_1,\ldots,U_5,V_1,\ldots,V_5),
	\end{eqnarray*}
	where $g_{30}$ is given by~\eqref{eqn-g30} and $h_{15}$ is given by~\eqref{eqn-h-5+2k}.
\item $\psi$ is the permutation $\psi_{15}$ given by~\eqref{eqn-psi}, with 15 written as $15=p_0+q_0=7+8$.
\item The function $\sym{load}(K,\sym{IV})$ outputs $K||\sym{IV}||(10)^{16}$.
\item The function $\sym{initG}$ is used for state initialisation.
\end{compactenum}

\paragraph{R-256.} The values of the parameters are as follows.
\begin{compactenum}
\item $\tau(x)=x^{256} \oplus x^{203} \oplus x^{138} \oplus x^{76} \oplus x^{46} \oplus x^{32} \oplus 1$, $A=(0,53,118,180,210,224)$;
\item $S_0=(7, 13, 19, 25, 31, 37, 43, 49, 55, 61, 67, 73, 79, 85, 91, 97, 103, 109,$ \\
	$217, 211, 205, 199, 193, 187, 181, 175, 169, 163, 157, 151, 145, 139, 133, 127, 121, 115);$ \\
	$S_1=(0,16, 26, 83, 84, 92, 134, 160, 192)$;
\item $P_0=(8, 74, 99, 131, 135, 136, 144, 189, 218)$, $P_1=(1,2,3,4,5,6)$;
\item $Q_0=(1, 11, 61, 110, 131, 133, 170, 198, 208, 218)$, $Q_1=(66, 74, 90, 97, 124, 193)$;
\item The functions $g$ and $h$ are as follows.
	\begin{eqnarray*}
		g(U_1,\ldots,U_{18},V_1,\ldots,V_{18}) & = & g_{36}(U_1,\ldots,U_{18},V_1,\ldots,V_{18}) \\
		h(X_1,X_2,Z_1,Z_2,Z_3,U_1,\ldots,U_7,V_1,\ldots,V_7) & = & h_{19}(X_1,X_2,Z_1,Z_2,Z_3,U_1,\ldots,U_7,V_1,\ldots,V_7),
	\end{eqnarray*}
	where $g_{36}$ is given by~\eqref{eqn-g36} and $h_{19}$ is given by~\eqref{eqn-h-5+2k}.
\item $\psi$ is the permutation $\psi_{19}$ given by~\eqref{eqn-psi} with 19 written as $19=p_0+q_0=9+10$.
\item The function $\sym{load}(K,\sym{IV})$ outputs $K||\sym{IV}||(10)^{32}$.
\item The function $\sym{initG}$ is used for state initialisation.
\end{compactenum}

\paragraph{W-256.} 
\begin{compactenum}
\item $\tau(x)=x^{208} \oplus x^{169} \oplus x^{164} \oplus x^{114} \oplus x^{35} \oplus x^{32} \oplus 1$, $A=(0,39,44,94,173,176)$;
\item $S_0=(7, 13, 19, 25, 31, 37, 43, 49, 55, 61, 67, 73, 79, 85, 91, 97, 103, 109,$ \\
	$217, 211, 205, 199, 193, 187, 181, 175, 169, 163, 157, 151, 145, 139, 133, 127, 121, 115);$ \\
	$S_1=(0,17, 38, 41, 89, 132, 146, 186, 190)$;
\item $P_0=(8, 72, 75, 99, 128, 176, 188, 212, 215)$, $P_1=(1,2,3,4,5,6)$;
\item $Q_0=(22, 53, 54, 73, 82, 86, 99, 143, 148, 167)$, $Q_1=(8, 70, 118, 151, 157, 171)$;
\item The functions $g$ and $h$ are as follows.
	\begin{eqnarray*}
		g(U_1,\ldots,U_{18},V_1,\ldots,V_{18}) & = & g_{36}(U_1,\ldots,U_{18},V_1,\ldots,V_{18}) \\
		h(X_1,X_2,Z_1,Z_2,Z_3,U_1,\ldots,U_7,V_1,\ldots,V_7) & = & h_{19}(X_1,X_2,Z_1,Z_2,Z_3,U_1,\ldots,U_7,V_1,\ldots,V_7),
	\end{eqnarray*}
	where $g_{36}$ is given by~\eqref{eqn-g36} and $h_{19}$ is given by~\eqref{eqn-h-5+2k}.
\item $\psi$ is the permutation $\psi_{19}$ given by~\eqref{eqn-psi} with 19 written as $19=p_0+q_0=9+10$.
\item The function $\sym{load}(K,\sym{IV})$ outputs $K||\sym{IV}||(10)^{8}$.
\item The function $\sym{initG}$ is used for state initialisation.
\end{compactenum}

\begin{table}
	{\scriptsize
	\centering
	\renewcommand\arraystretch{1.2}
	\begin{tabular}{|l|c|c|c|c|c|c|c|}
		\cline{3-8}
		\multicolumn{2}{c|}{ } & var & res & deg & AI & nl & LB \\ \hline
		\multirow{4}*{Grain~v1} & $g$ & 10 & $-1$ & 6 & 4 & 430 & $2^{-2.642}$ \\ \cline{2-8}
					& $G$ & 13 & 2 & 6 & 4 & $2^3\cdot 430$ & $2^{-2.642}$ \\ \cline{2-8}
					& $h$ & 5 & 1 & 3 & 2 & 12 & $2^{-2}$ \\ \cline{2-8}
					& $H$ & 12 & 8 & 3 & 3 & $2^7\cdot 12$ & $2^{-2}$ \\ \hline \hline
		\multirow{4}*{Grain-128a} & $g$ & 24 & $-1$ & 4 & $\leq 4$ & 8356352 & $2^{-8.023}$ \\ \cline{2-8}
					  & $G$ & 29 & 4 & 4 & $\leq 4$ & $2^5\cdot 8356352$ & $2^{-8.023}$ \\ \cline{2-8}
					  & $h$ & 9 & $-1$ & 3 & 3 & 240 & $2^{-4}$ \\ \cline{2-8}
					  & $H$ & 17 & 7 & 3 & 3 & $2^8\cdot 240$ & $2^{-4}$ \\ \hline\hline
		\multirow{4}*{R-80}    & $g$ & 10 & $-1$ & 7 & 4 & 492 & $2^{-4.678}$ \\ \cline{2-8}
					  & $G$ & 13 & 2 & 7 & 4 & $2^3\cdot 492$ & $2^{-4.678}$ \\ \cline{2-8}
					  & $h$ & 7 & 1 & 4 & 3 & 56 & $2^{-3}$ \\ \cline{2-8}
					  & $H$ & 14 & 8 & 3 & 3 & $2^7\cdot 56$ & $2^{-3}$ \\ \hline\hline
		\multirow{4}*{\begin{tabular}{l}R-128 \\ W-128 \end{tabular}} 
					  & $g$ & 24 & $-1$ & 6 & $\geq 4$ & $2^{23}-2^{11}$ & $2^{-12}$ \\ \cline{2-8}
					  & $G$ & 29 & 4 & 6 & $\geq 4$ & $2^{28}-2^{16}$ & $2^{-12}$ \\ \cline{2-8}
					  & $h$ & 10 & $-1$ & 5 & 3 & 496 & $2^{-5}$ \\ \cline{2-8}
					  & $H$ & 18 & 7 & 5 & 3 & $2^8\cdot 496$ & $2^{-5}$ \\ \hline\hline
		\multirow{4}*{\begin{tabular}{l}R-192 \\ W-192 \end{tabular}} 
					  & $g$ & 30 & $-1$ & 5 & $\geq 5$ & $2^{29}-2^{14}$ & $2^{-15}$ \\ \cline{2-8}
					  & $G$ & 37 & 6 & 5 & $\geq 5$ & $2^{35}-2^{23}$ & $2^{-15}$ \\ \cline{2-8}
					  & $h$ & 15 & 1 & 5 & 4 & $2^{14}-2^{7}$ & $2^{-7}$ \\ \cline{2-8}
					  & $H$ & 25 & 11 & 5 & $\geq 4$ & $2^{24}-2^{17}$ & $2^{-7}$ \\ \hline\hline
		\multirow{4}*{\begin{tabular}{l}R-256 \\ W-256 \end{tabular}} 
					  & $g$ & 36 & $-1$ & 8 & $\geq 6$ & $2^{35}-2^{17}$ & $2^{-18}$ \\ \cline{2-8}
					  & $G$ & 45 & 8 & 8 & $\geq 6$ & $2^{44}-2^{26}$ & $2^{-18}$ \\ \cline{2-8}
					  & $h$ & 19 & 1 & 7 & 4 & $2^{18}-2^{9}$ & $2^{-9}$ \\ \cline{2-8}
					  & $H$ & 31 & 13 & 7 & $\geq 4$ & $2^{30}-2^{21}$ & $2^{-9}$ \\ \hline\hline
	\end{tabular}
	\centering
	\caption{Comparison of cryptographic properties of the nonlinear functions used in the different proposals. \label{tab-cmp} }
	}
\end{table}
\begin{table}
        {\scriptsize
        \centering
	\renewcommand\arraystretch{1.2}
	\begin{tabular}{|l|c||c|c|c|c|c|}
		\cline{2-7}
		\multicolumn{1}{c|}{\ } & \#FFs & $g$ & $G$ & $h$ & $H$ & ($G+H$) \\ \hline
		Grain~v1 & 160 & 6[N]+12[X]+17[A] & 6[N]+15[X]+17[A] & 1[N]+7[X]+6[A] & 1[N]+14[X]+6[A] & 7[N]+29[X]+23[A] \\ \hline 
		Grain-128a & 256 & 9[X]+14[A] & 14[X]+14[A] & 4[X]+6[A] & 12[X]+6[A] & 26[X]+20[A] \\ \hline\hline
		R-80    & 160 & 1[N]+6[X]+10[A] & 1[N]+9[X]+10[A] & 2[N]+9[X]+8[A] & 2[N]+16[X]+8[A] & 3[N]+25[X]+18[A] \\ \hline
		R-128   & 256 & 1[N]+14[X]+17[A] & 1[N]+19[X]+17[A] & 5[X]+8[A] & 13[X]+8[A] & 1[N]+32[X]+25[A] \\ \hline
		W-128  & 240 & 1[N]+14[X]+17[A] & 1[N]+19[X]+17[A] & 5[X]+8[A] & 13[X]+8[A] & 1[N]+32[X]+25[A] \\ \hline
		R-192   & 384 & 1[N]+19[X]+21[A] & 1[N]+26[X]+21[A] & 12[X]+13[A] & 22[X]+13[A] & 1[N]+48[X]+34[A] \\ \hline
		W-192  & 352 & 1[N]+19[X]+21[A] & 1[N]+26[X]+21[A] & 12[X]+13[A] & 22[X]+13[A] & 1[N]+48[X]+34[A] \\ \hline
		R-256   & 512 & 1[N]+22[X]+27[A] & 1[N]+31[X]+27[A] & 14[X]+17[A] & 26[X]+17[A] & 1[N]+55[X]+44[A] \\ \hline
		W-256  & 464 & 1[N]+22[X]+27[A] & 1[N]+31[X]+27[A] & 14[X]+17[A] & 26[X]+17[A] & 1[N]+55[X]+44[A] \\ \hline
	\end{tabular}
	\caption{Comparison of the number of flip-flops and gate counts for the nonlinear functions used in the different proposals. \label{tab-gc-cmp} }
	}
\end{table}

\paragraph{Comparison and discussion.}
For the various proposals, the cryptographic properties of the constituent Boolean functions are given in Table~\ref{tab-cmp} and the corresponding
gate counts are given in Table~\ref{tab-gc-cmp}. Along with the details of the new proposals, we also provide the details of Grain~v1 and Grain-128a.

At the 80-bit security level, from Table~\ref{tab-cmp} we note that compared to Grain~v1, R-80 improves the algebraic degree (from 6 to 7) and 
linear bias (from $2^{-2.642}$ to $2^{-4.678}$) of $g$ while maintaining the same level of algebraic immunity, and the algebraic degree (from 3 to 4), algebraic 
immunity (from 2 to 3), and linear bias (from $2^{-2}$ to $2^{-3}$) of $h$. From Table~\ref{tab-gc-cmp}, we note that compared to Grain~v1, R-80 reduces the total gate count 
for $G$ and $H$ from 7[N]+29[X]+23[A] gates to 3[N]+25[X]+18[A] gates. So overall, R-80 improves the cryptographic properties of the constituent functions \textit{and} at the same
time reduces the gate counts. 

At the 128-bit security level, from Table~\ref{tab-cmp}, we note that compared to Grain-128a, R-128 and W-128 improve the algebraic degree (from 4 to 6) and the
linear bias (from $2^{-8.023}$ to $2^{-12}$) of $g$ without reducing the algebraic immunity, and improve the algebraic degree (from 3 to 5) and the linear bias 
(from $2^{-4}$ to $2^{-5}$) of $h$ while maintaining the same level of algebraic immunity.
From Table~\ref{tab-gc-cmp}, we note that compared to Grain-128a, R-128 and W-128 increase the total gate count for $G$ and $H$
from 26[X]+20[A] gates to 1[N]+32[X]+25[A] gates. So the improvement of the cryptographic properties of $g$ and $h$ offered by R-128 and W-128 come at a small increase
in gate count by 1[N]+6[X]+5[A] gates. While both Grain-128a and R-128 use 256 flip-flops, W-128 uses 240 flip-flops, i.e. a reduction of 16 flip-flops. Comparing Grain-128a 
and W-128, the reduction in the number of flip-flops by W-128 more than balances out the small increase in the number of gates required by $G$ and $H$. 
So overall W-128 has a smaller gate count than Grain-128a.

For the 192-bit and the 256-bit security levels, to the best of our knowledge, the gate counts of the proposals R-192, W-192, R-256 and W-256 are smaller than
all previous proposals for stream ciphers at these security levels.

\paragraph{Length of keystream.} 
We propose that at most $2^{64}$ keystream bits should be generated from a single pair of secret key and IV. See Appendix~\ref{app-bnd-keystream-mem}
for a justification of this bound.

We note that Grain~v1 and Grain-128a do not put any limit on the number of keystream bits that can be generated from a single key and IV pair.
The fast correlation attack~\cite{DBLP:conf/crypto/TodoIMAZ18} on Grain~v1 and Grain-128a requires $2^{75.1}$ and $2^{113.8}$ bits of keystream respectively 
to be generated from a single key and IV pair. As a response to this attack, Grain-128AEADv2 puts the restriction that at most $2^{81}$ bits should be generated from 
a single key-IV pair. 

\subsection{Security Conjecture \label{subsec-sec-conj}}
From a theoretical point of view, a stream cipher supporting an IV is modelled as pseudo-random function (PRF) that maps the IV to the keystream using the unknown secret 
key~\cite{DBLP:conf/fse/BerbainG07}. 
Assuming a stream cipher to be a PRF permits using the stream cipher for constructing other primitives and higher level functionalities. We mention some examples.
Certain kinds of almost XOR universal (AXU) hash functions require a key which is as long as the message (and even longer than the message for the Toeplitz
construction). A stream cipher can be used to generate such a key. Bit oriented construction of such hash functions was put forward in~\cite{DBLP:journals/dcc/Sarkar13}
and their implementation was reported in~\cite{DBLP:journals/tc/ChakrabortyMS15}. 
Generic methods to construct AEAD and deterministic AEAD (DAEAD) schemes based on stream ciphers and appropriate AXU hash functions were described 
in~\cite{DBLP:journals/ccds/Sarkar14} and these schemes were later analysed in~\cite{DBLP:journals/ijisec/ImamuraMI18}. Construction of a tweakable enciphering
scheme from a stream cipher was described in~\cite{DBLP:journals/iacr/Sarkar09f}, and hardware implementations of the scheme based on AXU functions
from~\cite{DBLP:journals/dcc/Sarkar13} and the stream ciphers in the hardware profile of the eStream recommendations were reported in~\cite{DBLP:journals/tc/ChakrabortyMS15}.

The adversarial model for the PRF security of stream ciphers is the following. Fix a stream cipher.
Suppose that at most $2^B$ keystream bits are to be generated from a single key and IV pair. (Recall that in our case $B=64$.)
Let $\mathcal{O}_0$ be an oracle which behaves as follows: a uniform random $\kappa$-bit secret key $K$ is chosen and the stream cipher is instantiated with $K$; on provided with 
a $v$-bit IV as input and a positive integer $\ell\leq 2^B$, $\mathcal{O}_0$ returns the first $\ell$-bit segment of the keystream generated using $K$ and IV. 
Let $\mathcal{O}_1$ be an oracle which on provided with a $v$-bit IV and a positive integer $\ell\leq 2^B$ as input returns
an independent and uniform random string of length $\ell$. An adversary is an algorithm $\mathcal{A}^{\mathcal{O}}$ which has access to an oracle $\mathcal{O}$;
$\mathcal{A}$ adaptively queries $\mathcal{O}$ with pairs $({\rm IV},\ell)$, where ${\rm IV}$ is a $v$-bit string and $\ell\leq 2^B$ is a positive integer, and receives
in return a bit string of length $\ell$. After querying the oracle $\mathfrak{q}$ times, $\mathcal{A}$ outputs a bit. The advantage of $\mathcal{A}$ is defined
to be $\lvert \Pr[\mathcal{A}^{\mathcal{O}_0} \mbox{ outputs } 1] - \Pr[\mathcal{A}^{\mathcal{O}_1} \mbox{ outputs } 1] \rvert$.
An alternative way to view the advantage of an adversary $\mathcal{A}$ is the following. A uniform random bit $b$ is chosen and $\mathcal{A}$ is given access to $\mathcal{O}_b$.
At the end of the interaction with $\mathcal{O}_b$, $\mathcal{A}$ outputs a bit $b^\prime$. It is not difficult to show that the absolute value of the bias of the bit
$b\oplus b^\prime$ is the advantage of $\mathcal{A}$. 

The resources of $\mathcal{A}$ consist of the following:
the total number $\mathfrak{q}$ of queries that it makes, the total number of bits $\mathfrak{B}$ that it receives in response to all its queries, the total time $\mathfrak{T}$ 
that it
takes, and the maximum number $\mathfrak{M}$ of bits of memory that it requires.
Since there are $2^v$ possible IV's, the maximum value of $\mathfrak{q}$ is $2^v$. We make the reasonable assumptions that $\mathfrak{M}\leq \mathfrak{T}$ 
(i.e. each memory position is accessed at least once), and $\mathfrak{B}\leq \mathfrak{T}$ (i.e. each bit received from the oracle is read at least once). 


Attacks which require physically impossible amounts of memory are perhaps not of any practical interest. In Appendix~\ref{app-bnd-keystream-mem} we
provide justifications for considering attacks requiring about $2^{100}$ or more bits of memory to be physically impossible to mount.


Suppose the advantage of an adversary is $\varepsilon$. By repeating the attack $\varepsilon^{-2}$ times, the advantage can be made close to 1. This increases
the required time by also a factor of $\varepsilon^{-2}$. Also, in each repetition the adversary is attacking a uniform and independently chosen random key, so
repetitions correspond to multi-target attacks. Repetitions, on the other hand, do not increase the memory requirement. 
Our security conjecture is the following.
\begin{conjecture}\label{conj-sec}
	For a stream cipher proposed in this work which is targeted at the $\kappa$-bit security level, and for which at most $2^{64}$ keystream bits are to be
	generated from a single pair of key and IV, for any adversary having advantage $\varepsilon$,
	requiring time $\mathfrak{T}$ and memory $\mathfrak{M}<2^{\min(\kappa,100)}$ bits, we conjecture that $\mathfrak{T}\cdot \varepsilon^{-2} < 2^\kappa$.
\end{conjecture}

\section{Analysis of Known Attacks \label{sec-attacks} }
Several classes of attacks have been proposed on the previously proposed members of the Grain family of stream ciphers. In this section, we provide descriptions
of these attacks and discuss their applicability to the new proposals.

\subsection{Exhaustive Search \label{subsec-ex-srch}}
The basic attack is to search exhaustively for the secret key. Given a segment of the keystream generated with an unknown key and a known IV, the attacker tries out
all possible choices of the secret key. For each choice of the secret key, the attacker generates a keystream segment of length equal to the given segment. If the generated
keystream segment is equal to the given keystream segment, then the corresponding choice of the secret key is a candidate key. If the keystream segment is long enough (a 
little longer than $\kappa$), then it is likely that there will be unique candidate key. Clearly the procedure requires trying out $2^{\kappa}$ candidate keys. 

\subsection{Linear Approximations \label{subsec-la} }
Linear approximations of a sum of keystream bits have been used earlier to attack members of the Grain family~\cite{DBLP:conf/ccs/Maximov06,DBLP:conf/fse/BerbainGM06}.
A comprehensive fast correlation attack on Grain~v1, Grain-128 and Grain-128a was proposed in~\cite{DBLP:conf/crypto/TodoIMAZ18}. 

The crux of all correlation attacks lies in obtaining ``good'' linear approximations of a sum of keystream bits in terms of the bits of the LFSR sequence. 
This allows a divide-and-conquer attack, where the LFSR state is first recovered and then the NFSR state is recovered~\cite{DBLP:conf/crypto/TodoIMAZ18}.
The task of obtaining linear approximations was done separately for Grain~v1, Grain-128, and Grain-128a in~\cite{DBLP:conf/crypto/TodoIMAZ18}. 
In this section, we perform a very general and theoretical analysis of bias of linear approximations for the entire abstract Grain family that we
have defined. For this purpose, we will use some results on correlations which are given in Appendix~\ref{app-corr}. 

As in Section~\ref{sec-grain-family}, let $(N^{(0)},L^{(0)})$ be the state produced after the initialisation phase (Steps~\ref{step-load} and~\ref{step-init}) of 
Algorithm~\ref{algo-grain}. For $t\geq 0$, let $N^{(t)}=(\eta_t,\eta_{t+1},\ldots,\eta_{t+\kappa-1})$, $L^{(t)}=(\lambda_t,\lambda_{t+1},\ldots,\lambda_{t+\kappa-1})$, 
and the $t$-th keystream bit be $z_t$. 

We introduce a new notation. For $t\geq 0$ and $w\geq 1$, let 
\begin{eqnarray}\label{eqn-lambda-eta-window}
	\bm{\lambda}_{t,w}=(\lambda_t,\lambda_{t+1},\ldots,\lambda_{t+w-1}) \mbox{ and } \bm{\eta}_{t,w}=(\eta_t,\eta_{t+1},\ldots,\eta_{t+w-1}).
\end{eqnarray}
Let $T$ be a finite non-empty set of non-negative integers, and 
\begin{eqnarray}\label{eqn-S1-prime}
	S_1^\prime & = & T \triangle (S_1\cup \{\kappa_1\}), 
\end{eqnarray}
where $\triangle$ denotes symmetric difference of sets. Using~\eqref{eqn-nnb} we can write
\begin{eqnarray*}
\lambda_t \oplus g(\sym{proj}(S_0,N^{(t)})) 
	& = & \eta_{t+\kappa_1} \oplus \sym{xor}(\sym{proj}(S_1,N^{(t)})) \nonumber \\
	& = & \left(\bigoplus_{j\in S_1\cup\{\kappa_1\}}\eta_{t+j}\right) \oplus \left(\bigoplus_{i\in T}\eta_{t+i} \right) \oplus \left( \bigoplus_{i\in T}\eta_{t+i} \right) \\
	& = & \left(\bigoplus_{i\in S_1^\prime} \eta_{t+i} \right) \oplus \left( \bigoplus_{i\in T}\eta_{t+i} \right). 
\end{eqnarray*}
So 
\begin{eqnarray}
	\lambda_t \oplus g(\sym{proj}(S_0,N^{(t)})) \oplus \bigoplus_{i\in S_1^\prime} \eta_{t+i} & = & \bigoplus_{i\in T}\eta_{t+i} \label{eqn-S1prime}
\end{eqnarray}

From~\eqref{eqn-op-bit}, we have
\begin{eqnarray}
	\bigoplus_{i\in T} z_{t+i}
	& = & \bigoplus_{i\in T} \sym{xor}(\sym{proj}(Q_1,L^{(t+i)}))
	      \oplus \bigoplus_{i\in T} \sym{xor}(\sym{proj}(P_1,N^{(t+i)})) \nonumber \\
	&   & \oplus \bigoplus_{i\in T} h(\psi(\sym{proj}(Q_0,L^{(t+i)}),\sym{proj}(P_0,N^{(t+i)}))) \nonumber \\
	& = & \bigoplus_{i\in T} \bigoplus_{j\in Q_1} \lambda_{t+i+j} \oplus \bigoplus_{i\in T} \bigoplus_{j\in P_1} \eta_{t+i+j} 
	      \oplus \bigoplus_{i\in T} h(\psi(\sym{proj}(Q_0,L^{(t+i)}),\sym{proj}(P_0,N^{(t+i)}))) \nonumber \\
	& = & \bigoplus_{i\in T} \bigoplus_{j\in Q_1} \lambda_{t+i+j} \oplus \bigoplus_{j\in P_1} \bigoplus_{i\in T} \eta_{t+i+j} 
	      \oplus \bigoplus_{i\in T} h(\psi(\sym{proj}(Q_0,L^{(t+i)}),\sym{proj}(P_0,N^{(t+i)}))) \nonumber \\
	&   & \quad\quad\quad\quad\quad\quad\quad\quad\quad\quad\quad\quad\quad\quad\quad \mbox{(by interchanging the sums over $P_1$ and $T$)} \nonumber \\
	& = & \bigoplus_{i\in T} \bigoplus_{j\in Q_1} \lambda_{t+i+j} 
	      \oplus \bigoplus_{j\in P_1} 
	      \left( \lambda_{t+j} \oplus g(\sym{proj}(S_0,N^{(t+j)})) \oplus \bigoplus_{i\in S_1^\prime} \eta_{t+i+j} \right) \quad \mbox{(from~\eqref{eqn-S1prime})} \nonumber \\
	&   & \oplus \bigoplus_{i\in T} h(\psi(\sym{proj}(Q_0,L^{(t+i)}),\sym{proj}(P_0,N^{(t+i)})))) \nonumber \\
	& = & \bigoplus_{i\in T} \bigoplus_{j\in Q_1} \lambda_{t+i+j} \oplus \bigoplus_{j\in P_1} \lambda_{t+j}  
	      \oplus \bigoplus_{j\in P_1} g(\sym{proj}(S_0,N^{(t+j)})) \oplus \bigoplus_{j\in P_1}\bigoplus_{i\in S_1^\prime}\eta_{t+i+j} \nonumber \\
	&   & \oplus \bigoplus_{i\in T} h(\psi(\sym{proj}(Q_0,L^{(t+i)}),\sym{proj}(P_0,N^{(t+i)}))). \label{eqn-t0}
\end{eqnarray}
Let 
\begin{eqnarray}\label{eqn-r-s}
	\left.
	\begin{array}{rcl}
		r & = & 1+\max(\max(Q_1+T), \max(P_1), \max(Q_0+T)), \\
		s & = & 1+\max(\max(P_1+S_0), \max(P_1+S_1^{\prime}), \max(P_0+T)). 
	\end{array} \right\}
\end{eqnarray}
Consider the sequences $\bm{\lambda}_{t,r}$ (of length $r$) and $\bm{\eta}_{t,s}$ (of length $s$). In~\eqref{eqn-t0} no bit outside the sequences
$\bm{\lambda}_{t,r}$ and $\bm{\eta}_{t,s}$ are involved. So we may focus our attention only on the bits of $\bm{\lambda}_{t,r}$ and $\bm{\eta}_{t,s}$. Note that not all of 
the bits of these two sequences appear in~\eqref{eqn-t0}.

Let $\bm{\gamma}\in\mathbb{F}_2^r$ and define the bit 
\begin{eqnarray}\label{eqn-b-t}
	b_{t,T,\bm{\gamma}} & = & \langle \bm{\gamma}, \bm{\lambda}_{t,r}\rangle \oplus \bigoplus_{i\in T} z_{t+i}.
\end{eqnarray}
For a fixed $T$, we are interested in all those $\bm{\gamma}$'s such that $\lvert \sym{bias}(b_{t,T,\bm{\gamma}})\rvert$ is sufficiently high for all $t\geq 0$. 
The fast correlation attack in~\cite{DBLP:conf/crypto/TodoIMAZ18} requires about $\lvert \sym{bias}(b_{t,T,\bm{\gamma}})\rvert^{-2}$ keystream bits.

Using~\eqref{eqn-t0}, we have
\begin{eqnarray}\label{eqn-b-t-long}
	b_{t,T,\bm{\gamma}} 
	& = & \langle \bm{\gamma}, \bm{\lambda}_{t,r}\rangle \oplus 
	\bigoplus_{i\in T} \bigoplus_{j\in Q_1} \lambda_{t+i+j} \oplus \bigoplus_{j\in P_1} \lambda_{t+j} 
	\oplus \bigoplus_{i\in T} h(\psi(\sym{proj}(Q_0,L^{(t+i)}),\sym{proj}(P_0,N^{(t+i)}))) \nonumber \\
	& & \oplus \bigoplus_{j\in P_1}\bigoplus_{i\in S_1^\prime}\eta_{t+i+j} \oplus \bigoplus_{j\in P_1} g(\sym{proj}(S_0,N^{(t+j)})).
\end{eqnarray}

Given a non-empty subset $T$ of non-negative integers and $\bm{\gamma}\in\mathbb{F}_2^r$, define
\begin{eqnarray} \label{eqn-BCDEF}
	\left.\begin{array}{l}
B = (T+Q_1) \triangle P_1 \triangle \sym{supp}(\bm{\gamma}). \\
	C = P_1 + S_1^{\prime}$, \mbox{ where } $S_1^{\prime}=T\triangle (S_1\cup \{\kappa_1\}).  \\
	\mbox{For } i\in T, D_i=Q_0+i, \mbox{ and } D=\cup_{i\in T}D_i = Q_0+T. \\
	\mbox{For } i\in T, E_i=P_0+i, \mbox{ and } E=\cup_{i\in T}E_i = P_0+T.  \\
	\mbox{For } j\in P_1, F_j=S_0+j, \mbox{ and } F=\cup_{j\in P_1}F_j=S_0+P_1. 
	\end{array} \right\}
\end{eqnarray}
With the above notation and the definitions of $\bm{\lambda}_{t,r}$ and $\bm{\eta}_{t,s}$, for each $i\in T$ and $j\in P_1$, we have
\begin{eqnarray}\label{eqn-proj}
	\left.
	\begin{array}{rcl}
		\sym{proj}(D_i,\bm{\lambda}_{t,r}) & = & \sym{proj}(Q_0,L^{(t+i)}), \\
		\sym{proj}(E_i,\bm{\eta}_{t,s}) & = & \sym{proj}(P_0,N^{(t+i)}), \\
		\sym{proj}(F_j,\bm{\eta}_{t,s}) & = & \sym{proj}(S_0,N^{(t+j)}).
	\end{array} \right\}
\end{eqnarray}
Using~\eqref{eqn-BCDEF} and~\eqref{eqn-proj}, we can rewrite $b_{t,T,\bm{\gamma}}$ given by~\eqref{eqn-b-t-long} as follows.
\begin{eqnarray}\label{eqn-b-t-long-another}
	b_{t,T,\bm{\gamma}} 
	& = & \sym{xor}(\sym{proj}(B,\bm{\lambda}_{t,r})) \oplus \bigoplus_{i\in T} h(\psi(\sym{proj}(D_i,\bm{\lambda}_{t,r}),\sym{proj}(E_i,\bm{\eta}_{t,s}))) \nonumber \\
	& & \quad\quad\quad\quad \oplus \sym{xor}(\sym{proj}(C,\bm{\eta}_{t,s})) \oplus \bigoplus_{j\in P_1} g(\sym{proj}(F_j,\bm{\eta}_{t,s})).
\end{eqnarray}
\begin{proposition}\label{prop-zero-corr}
	If $B\setminus D\neq \emptyset$, or $C\setminus (E\cup F)\neq \emptyset$, then $\sym{bias}(b_{t,T,\bm{\gamma}})=0$.
	Consequently, if $b_{t,T,\bm{\gamma}}$ has a non-zero bias, then 
\begin{eqnarray}\label{eqn-cond}
(T+Q_1) \triangle P_1 \triangle \sym{supp}(\bm{\gamma}) \subseteq Q_0+T & \mbox{and} & P_1 + (T\triangle (S_1\cup \{\kappa_1\})) \subseteq (P_0+T) \cup (S_0+P_1).
\end{eqnarray}
\end{proposition}
\begin{proof}
	The positions of $B$ determine the bits of $\bm{\lambda}_{t,r}$ which appear as linear terms in~\eqref{eqn-b-t-long-another}, while the positions of $D$ determine
	the bits of $\bm{\lambda}_{t,r}$ which appear in the inputs to the different invocations of $h$ in~\eqref{eqn-b-t-long-another}. If there is a position in $B$ which
	does not appear in $D$, then the entire expression given by~\eqref{eqn-b-t-long} is uniformly distributed (due to the assumption on the independent and uniform
	distribution of the bits of $\bm{\lambda}_{t,r}$) and hence $\sym{bias}(b_{t,T,\bm{\gamma}})=0$. 
	The statement that if $C\setminus (E\cup F)\neq \emptyset$, then $\sym{bias}(b_{t,T,\bm{\gamma}})=0$ is obtained by a similar argument over $\bm{\eta}_{t,s}$.

	So if $\sym{bias}(b_{t,T,\bm{\gamma}})\neq 0$, then $B\subseteq D$ and $C\subseteq E\cup F$. Then~\eqref{eqn-cond} is obtained using the
	definitions of $B$, $C$, $D$, $E$ and $F$ given in~\eqref{eqn-BCDEF}.
\end{proof}

The condition on $T$ given by~\eqref{eqn-cond} is a necessary condition for $b_{t,T,\bm{\gamma}}$ to have a nonzero bias. While choosing $T$, the
cryptanalyst must ensure that~\eqref{eqn-cond} holds, as otherwise the bias of $b_{t,T,\bm{\gamma}}$ will be zero. So Proposition~\ref{prop-zero-corr} provides
a basic guidance on how to choose $T$. Note that since~\eqref{eqn-cond} is a necessary condition, it is not guaranteed that choosing $T$ satisfying this condition ensures
that $b_{t,T,\bm{\gamma}}$ has a non-zero bias.

Let $\bm{\gamma}^\prime\in \mathbb{F}_2^r$ and $\bm{\delta}\in \mathbb{F}_2^s$ be such that
\begin{eqnarray} \label{eqn-gamma-delta}
	\langle \bm{\gamma}^\prime, \bm{\lambda}_{t,r}\rangle = \sym{xor}(\sym{proj}(B,\bm{\lambda}_{t,r}))
	& \mbox{and} & 
	\langle \bm{\delta}, \bm{\eta}_{t,s}\rangle = \sym{xor}(\sym{proj}(C,\bm{\eta}_{t,s})).
\end{eqnarray}
Further, for $i\in T$ and for $j\in P_1$ define $\mathfrak{h}_i$ and $\mathfrak{g}_j$ to be functions such that 
\begin{eqnarray} \label{eqn-hi-gj}
	\left.
	\begin{array}{rcl}
	\mathfrak{h}_i(\bm{\lambda}_{t,r},\bm{\eta}_{t,s}) & = & 
	h(\psi(\sym{proj}(D_i,\bm{\lambda}_{t,r}),\sym{proj}(E_i,\bm{\eta}_{t,s}))) \\
	\mathfrak{g}_j(\bm{\eta}_{t,s}) & = & 
		g(\sym{proj}(F_j,\bm{\eta}_{t,s})).
	\end{array}\right\}
\end{eqnarray}
From the definition, it follows that $\mathfrak{h}_i$ is degenerate on all bits of $\bm{\lambda}_{t,r}$ other than those indexed by $D_i$, 
$\mathfrak{h}_i$ is degenerate on all bits of $\bm{\eta}_{t,s}$ other than those indexed by $E_i$, and $\mathfrak{g}_j$ is degenerate on all bits of 
$\bm{\eta}_{t,s}$ other than those indexed by $F_j$. 

Let $\mathfrak{h}$ and $\mathfrak{g}$ be functions such that 
\begin{eqnarray} \label{eqn-h-g}
	\mathfrak{h}(\bm{\lambda}_{t,r},\bm{\eta}_{t,s}) = \bigoplus_{i\in T} \mathfrak{h}_i(\bm{\lambda}_{t,r},\bm{\eta}_{t,s})
	& \mbox{and} & 
	\mathfrak{g}(\bm{\eta}_{t,s}) = \bigoplus_{j\in P_1} \mathfrak{g}_j(\bm{\eta}_{t,s}). 
\end{eqnarray}
Using~\eqref{eqn-gamma-delta},~\eqref{eqn-hi-gj} and~\eqref{eqn-h-g}, we rewrite~\eqref{eqn-b-t-long-another} as follows.
\begin{eqnarray}
	b_{t,T,\bm{\gamma}} 
	& = & \langle \bm{\gamma}^\prime, \bm{\lambda}_{t,r}\rangle \oplus \langle \bm{\delta}, \bm{\eta}_{t,s}\rangle
	\oplus \bigoplus_{i\in T} \mathfrak{h}_i(\bm{\lambda}_{t,r},\bm{\eta}_{t,s}) 
	\oplus \bigoplus_{j\in P_1} \mathfrak{g}_j(\bm{\eta}_{t,s}) \label{eqn-b-t-short1} \\
	& = & \langle \bm{\gamma}^\prime, \bm{\lambda}_{t,r}\rangle \oplus \langle \bm{\delta}, \bm{\eta}_{t,s}\rangle
	\oplus \mathfrak{h}(\bm{\lambda}_{t,r},\bm{\eta}_{t,s}) \oplus \mathfrak{g}(\bm{\eta}_{t,s}). \label{eqn-b-t-short2}
\end{eqnarray}
Using Propositions~\ref{prop-corr-degenerate} and~\ref{prop-corr-perm} (given in Appendix~\ref{app-corr}), we relate the correlations of $\mathfrak{h}_i$ and $\mathfrak{g}_j$ to 
those of $h$ and $g$ respectively.
\begin{proposition}\label{prop-h-i-g-j-corr}
For $\bm{\alpha}\in\mathbb{F}_2^r$ and $\bm{\beta}\in\mathbb{F}_2^s$,
\begin{eqnarray}\label{eqn-corr-h-i}
	\lefteqn{\sym{corr}_{\mathfrak{h}_i}(\bm{\alpha},\bm{\beta})} \nonumber \\
	& = & \left\{
		\begin{array}{lcl}
			0 & \mbox{if} & \sym{supp}(\bm{\alpha})\setminus D_i\neq \emptyset \mbox{ or } \sym{supp}(\bm{\beta})\setminus E_i\neq \emptyset, \\
			\sym{corr}_h(\psi(\sym{proj}(D_i,\bm{\alpha}),\sym{proj}(E_i,\bm{\beta})))
			& \mbox{if} & \sym{supp}(\bm{\alpha}) \subseteq D_i \mbox{ and } \sym{supp}(\bm{\beta})\subseteq E_i.
		\end{array}
		\right.
\end{eqnarray}
For $\bm{\beta}\in\mathbb{F}_2^s$,
\begin{eqnarray}\label{eqn-corr-g-j}
	\sym{corr}_{\mathfrak{g}_j}(\bm{\beta})
	& = & \left\{
		\begin{array}{lcl}
			0 & \mbox{if} & \sym{supp}(\bm{\beta})\setminus F_j\neq \emptyset, \\
			\sym{corr}_g(\sym{proj}(F_j,\bm{\beta})) & \mbox{if} & \sym{supp}(\bm{\beta})\subseteq F_j.
		\end{array}
		\right.
\end{eqnarray}
\end{proposition}

Using the definitions of $\mathfrak{h}$ and $\mathfrak{g}$ given by~\eqref{eqn-h-g} and Theorem~\ref{thm-gen-conv} (given in Appendix~\ref{app-corr}), we have the following result.
\begin{proposition}\label{prop-h-corr}
	Let $\mathfrak{t}=\#T$. Then
\begin{eqnarray}\label{eqn-corr-h}
	\sym{corr}_{\mathfrak{h}}(\bm{\alpha},\bm{\beta})
	& = & \sum_{\{(\mathbf{u}_i,\mathbf{v}_i)\in\mathbb{F}_2^{r}\times \mathbb{F}_2^s: 1\leq i\leq \mathfrak{t}-1\}} 
	\prod_{i=1}^{\mathfrak{t}} \sym{corr}_{\mathfrak{h}_i}(\mathbf{u}_i\oplus \mathbf{u}_{i-1},\mathbf{v}_i\oplus \mathbf{v}_{i-1}),
\end{eqnarray}
where $\mathbf{u}_0=\mathbf{0}_{r}$, and $\mathbf{u}_{\mathfrak{t}}=\bm{\alpha}$, $\mathbf{v}_0=\mathbf{0}_s$, and $\mathbf{v}_{\mathfrak{t}}=\bm{\beta}$. 

Further,
\begin{eqnarray}\label{eqn-corr-g}
        \sym{corr}_{\mathfrak{g}}(\bm{\beta})
        & = & \sum_{\{\mathbf{w}_j\in\mathbb{F}_2^s: 1\leq j\leq p_1-1\}} \prod_{j=1}^{p_1} \sym{corr}_{\mathfrak{g}_j} (\mathbf{w}_j\oplus \mathbf{w}_{j-1}),
\end{eqnarray}
where $\mathbf{w}_0=\mathbf{0}_{s}$, and $\mathbf{w}_{p_1}=\bm{\beta}$.
\end{proposition}

\begin{remark}\label{rem-corr-h-abs-val}
Note that the expression in~\eqref{eqn-corr-h} is a sum of products, where each product term is obtained
	by multiplying together $\mathfrak{t}$ correlations of $\mathfrak{h}_i$, which from~\eqref{eqn-corr-h-i} amounts to multiplying together
	at most $\mathfrak{t}$ correlations of $h$. So the absolute value of each of these products of correlations of $h$ is at most 
$\varepsilon_{h,\mathfrak{t}}=\sym{LB}(h)^{\mathfrak{t}}$. 
\end{remark}

\begin{theorem}\label{thm-la}
	Let $t\geq 0$, $T$ be a finite non-empty set of integers, and $r$ be defined as in~\eqref{eqn-r-s}.
	For $\bm{\gamma}\in\mathbb{F}_2^r$, let $b_{t,T,\bm{\gamma}}$ be defined as in~\eqref{eqn-b-t}. Let $\bm{\gamma}^\prime$ and $\bm{\delta}$ be defined
	as in~\eqref{eqn-gamma-delta}. Then
	\begin{eqnarray}
		\sym{bias}(b_{t,T,\bm{\gamma}}) & = & \sum_{\bm{\beta}\in \mathbb{F}_2^{r}} 
	\sym{corr}_{\mathfrak{h}}(\bm{\gamma}^{\prime},\bm{\beta})\cdot \sym{corr}_{\mathfrak{g}}(\bm{\beta}\oplus \bm{\delta}). \label{eqn-corr-h-g}
	\end{eqnarray}
\end{theorem}
\begin{proof}
Using~\eqref{eqn-b-t-short2}, we have 
\begin{eqnarray}
	b_{t,T,\bm{\gamma}} 
	& = & \langle (\bm{\gamma}^{\prime},\bm{\delta}), (\bm{\lambda}_{t,r}, \bm{\eta}_{t,s})\rangle
	\oplus \mathfrak{h}(\bm{\lambda}_{t,r},\bm{\eta}_{t,s}) \oplus \mathfrak{g}(\bm{\eta}_{t,s}). \label{eqn-t1}
\end{eqnarray}
So $\sym{bias}(b_{t,T,\bm{\gamma}})$ is equal to $\sym{corr}_{\mathfrak{f}}(\bm{\gamma}^{\prime},\bm{\delta})$, where 
$\mathfrak{f}(\bm{\lambda}_{t,r},\bm{\eta}_{t,s})=\mathfrak{h}(\bm{\lambda}_{t,r},\bm{\eta}_{t,s}) \oplus \mathfrak{g}(\bm{\eta}_{t,s})$.
	From Theorem~\ref{thm-gen-conv} (given in Appendix~\ref{app-corr}), we have
\begin{eqnarray}
	\sym{corr}_{\mathfrak{f}}(\bm{\gamma}^{\prime},\bm{\delta})
	& = & \sum_{(\bm{\alpha},\bm{\beta})\in \mathbb{F}_2^{s}\times \mathbb{F}_2^{r}} 
	\sym{corr}_{\mathfrak{h}}(\bm{\alpha},\bm{\beta})\cdot \sym{corr}_{\mathfrak{g}}((\bm{\alpha},\bm{\beta})\oplus (\bm{\gamma}^{\prime},\bm{\delta})). \label{eqn-t2}
\end{eqnarray}
Since $\mathfrak{g}$ is degenerate on $\bm{\lambda}_{t,r}$, from Proposition~\ref{prop-corr-degenerate} we have,
$\sym{corr}_{\mathfrak{g}}((\bm{\alpha},\bm{\beta})\oplus (\bm{\gamma}^{\prime},\bm{\delta}))=0$ unless $\bm{\alpha}=\bm{\gamma}^{\prime}$. So~\eqref{eqn-t2} simplifies to
\begin{eqnarray}
	\sym{corr}_{\mathfrak{f}}(\bm{\gamma}^{\prime},\bm{\delta})
	& = & \sum_{\bm{\beta}\in \mathbb{F}_2^{r}} 
	\sym{corr}_{\mathfrak{h}}(\bm{\gamma}^{\prime},\bm{\beta})\cdot \sym{corr}_{\mathfrak{g}}(\bm{\beta}\oplus \bm{\delta}). \label{eqn-t3}
\end{eqnarray}
\end{proof}

In~\eqref{eqn-corr-h-g}, the value of $\sym{corr}_{\mathfrak{h}}(\bm{\gamma}^{\prime},\bm{\beta})$ is obtained from~\eqref{eqn-corr-h} as a sum of
products of the correlations of the $\mathfrak{h}_i$'s, where as mentioned in Remark~\ref{rem-corr-h-abs-val}
the absolute value of each of these products is at most $\sym{LB}(h)^{\mathfrak{t}}$. 

From Table~\ref{tab-cmp}, $\sym{LB}(h)=2^{-2}$ and $\sym{LB}(h)=2^{-4}$ for Grain~v1 and Grain-128a respectively. 
In~\cite{DBLP:conf/crypto/TodoIMAZ18}, $\mathfrak{t}$ was taken to be 10 and 6 for Grain~v1 and Grain-128a respectively. The corresponding values
of $\varepsilon_{h,\mathfrak{t}}$ are $2^{-20}$ and $2^{-24}$ for Grain~v1 and Grain-128a respectively. 

Using the values of $\sym{LB}(h)$ (from Table~\ref{tab-cmp}), we have $\varepsilon_{h,\mathfrak{t}}$ equal to $2^{-3\mathfrak{t}}$,
$2^{-5\mathfrak{t}}$, $2^{-7\mathfrak{t}}$ and $2^{-9\mathfrak{t}}$ for the 80-bit (proposal R-80), 128-bit (proposals R-128 and W-128), 192-bit (proposals
R-192 and W-192), and 256-bit (proposals R-256 and W-256) security levels respectively. The actual value of $\varepsilon_{h,\mathfrak{t}}$ depends on the
value of $\mathfrak{t}$. If, as in~\cite{DBLP:conf/crypto/TodoIMAZ18}, we choose $\mathfrak{t}=10$ and $\mathfrak{t}=6$ for the 80-bit and the 128-bit security
levels, then the corresponding values of $\varepsilon_{h,\mathfrak{t}}$ that we obtain are both equal to $2^{-30}$ which is substantially lower than the values
of $\varepsilon_{h,\mathfrak{t}}$ obtained for Grain~v1 and Grain-128a. 

Next we consider $\sym{corr}_{\mathfrak{g}}$ which appears in~\eqref{eqn-corr-h-g}.
Combining~\eqref{eqn-h-g} and~\eqref{eqn-hi-gj} we see that the function $\mathfrak{g}$ is a sum of $p_1$ copies of the function $g$.
The correlation of $\mathfrak{g}$ can be expressed in terms of the correlations of $g$ by combining~\eqref{eqn-corr-g} and~\eqref{eqn-corr-g-j}. 
In~\cite{DBLP:conf/crypto/TodoIMAZ18} however, for both Grain~v1 and Grain-128a, the correlation of $\mathfrak{g}$ was estimated by directly considering its ANF and 
individually estimating the correlation of various groups of terms in the ANF. 
In general, even though $\mathfrak{g}$ is a sum of $p_1$ copies of the function $g$, the different copies of $g$ are not applied to disjoint set of variables, i.e.
$\mathfrak{g}$ is not a direct sum of $p_1$ copies of $g$. If $\mathfrak{g}$ were such a direct sum, then from Proposition~\ref{prop-dsum-corr} (in Appendix~\ref{app-corr}) 
it would have been easy to 
obtain the correlations of $\mathfrak{g}$ in terms of the correlations of $g$. Whether or not $\mathfrak{g}$ is a direct sum of $p_1$ copies of $g$ depends on the
choice of the tap positions. The following result provides a sufficient condition for $\mathfrak{g}$ to be a direct sum and in this case relates
the correlation of $\mathfrak{g}$ to the correlation of $g$. 


\begin{proposition}\label{prop-g-dsum-cor}
	If $\#(S_0+P_1)=\#S_0 \cdot \#P_1$, then for $\mathfrak{g}$ defined in~\eqref{eqn-h-g},
	\begin{eqnarray}\label{eqn-g-sp}
		\sym{corr}_{\mathfrak{g}}(\bm{\beta}) & = & \prod_{j\in P_1}\sym{corr}_{g}(\sym{proj}(F_j,\bm{\beta})).
	\end{eqnarray}
	Consequently, 
	\begin{eqnarray}\label{eqn-g-ub}
		\lvert \sym{corr}_{\mathfrak{g}}(\bm{\beta}) \rvert & \leq & \varepsilon_g.
	\end{eqnarray}
	where $\varepsilon_g=\sym{LB}(g)^{p_1}$.
\end{proposition}
\begin{proof}
	The condition $\#(S_0+P_1)=\#S_0 \cdot \#P_1$ holds if and only if the $F_j$'s are pairwise disjoint. 
	Using the definition of $\mathfrak{g}_j$ from~\eqref{eqn-hi-gj} in the defintion of $\mathfrak{g}$ in~\eqref{eqn-h-g}, we have
	\begin{eqnarray*}
		\mathfrak{g}(\bm{\eta}_{t,s})
		& = & \bigoplus_{j\in P_1} \mathfrak{g}_j(\bm{\eta}_{t,s}) = \bigoplus_{j\in P_1} g(\sym{proj}(F_j,\bm{\eta}_{t,s})).
	\end{eqnarray*}
	Since the $F_j$'s are pairwise disjoint it follows that $\mathfrak{g}$ is a direct sum of $p_1$ copies of $g$. 
	The result then follows from Proposition~\ref{prop-dsum-corr} (given in Appendix~\ref{app-corr}).

	Since $\lvert \sym{corr}_{g}(\sym{proj}(F_j,\bm{\beta}))\rvert \leq \sym{LB}(g)$, we obtain~\eqref{eqn-g-ub} from~\eqref{eqn-g-sp}.
\end{proof}
\begin{remark}\label{rem-dsum}
Using Proposition~\ref{prop-g-dsum-cor}, if we choose $S_0$ and $P_1$ such that $\#(S_0+P_1)=\#S_0 \cdot \#P_1$,
	then we obtain a good upper bound on the quantity $\lvert \sym{corr}_{\mathfrak{g}}(\bm{\beta}\oplus \bm{\delta}) \rvert$ appearing in~\eqref{eqn-corr-h-g}, 
	namely $\lvert \sym{corr}_{\mathfrak{g}}(\bm{\beta}\oplus \bm{\delta}) \rvert \leq \sym{LB}(g)^{p_1}$. So based on Proposition~\ref{prop-g-dsum-cor} we set a 
design condition that $\#(S_0+P_1)=\#S_0 \cdot \#P_1$. The new proposals that we put forward satisfy this condition. Using the values of $p_1$ 
	(from Table~\ref{tab-params}) and $\sym{LB}(g)$ (from Table~\ref{tab-cmp}), the upper bounds on $\lvert \sym{corr}_{\mathfrak{g}}(\bm{\beta}\oplus \bm{\delta}) \rvert$ 
	are $2^{-28.068}$, $2^{-48}$, $2^{-75}$ and $2^{-108}$ for the 80-bit (proposal R-80), 128-bit (proposals R-128 and W-128), 192-bit (proposals R-192 and W-192),
	and 256-bit (proposals R-256 and W-256) security levels respectively. 
\end{remark}

From a security point of view, it is desirable to be able to prove a good upper bound on $\lvert \sym{bias}(b_{t,T,\bm{\gamma}})\rvert$ as given by~\eqref{eqn-corr-h-g}.
While, as discussed above, we can prove that $\varepsilon_{h,\mathfrak{t}}$ and $\varepsilon_g$ are provably low (see Remarks~\ref{rem-corr-h-abs-val} and~\ref{rem-dsum}), 
due to the complications created by the signs
of the correlations it seems difficult to use these values to prove useful upper bounds on $\lvert \sym{bias}(b_{t,T,\bm{\gamma}})\rvert$. 
So the low values of $\varepsilon_{h,\mathfrak{t}}$ and $\varepsilon_g$ provide good evidence of resistance to fast correlation attacks, but do not provide
provable security against such attacks.

\subsection{Chosen IV Attacks \label{subsec-chosen-IV-attack}}
In chosen IV attacks, one or more bits of the keystream are obtained from the stream cipher with the same secret key and several chosen IVs. The obtained keystream
bits are analysed to gain information about the secret key. Two types of chosen IV attacks that have been applied to the Grain family are cube attacks and
conditional differential attacks. Below we briefly discuss these two types of attacks.

\subsubsection{Cube Attacks \label{subsubsec-cube-attack}}
Fix any bit $z$ of the keystream (produced after initialisation phase is over). This keystream bit is a function of the secret key $K$ and the initialisation vector IV. 
Let us denote this dependence as $z=f(K,{\rm IV})$, where $f$ is a Boolean function of $\kappa+v$ variables $K_1,\ldots,K_{\kappa}$ and $V_1,\ldots,V_v$.
Suppose there is a subset $\{i_1,\ldots,i_u\}\subseteq \{1,\ldots,v\}$, and corresponding list of variables $\mathbf{V}=(V_{i_1},\ldots,V_{i_u})$ such that 
\begin{eqnarray}
	f(K,{\rm IV}) & = & V_{i_1}\cdots V_{i_u} p(K,{\rm IV}) \oplus q(K,{\rm IV}), \label{eqn-cube-form}
\end{eqnarray}
for some Boolean functions $p(K,{\rm IV})$ and $q(K,{\rm IV})$ satisfying the following two conditions.
\begin{enumerate}
\item $p(K,{\rm IV})$ does not involve any variable from $\mathbf{V}$.
\item Each term in the ANF of $q(K,{\rm IV})$ misses out at least one of the variables from $\mathbf{V}$.
\end{enumerate}
For $\mathbf{v}\in\mathbb{F}_2^u$, let $f_{\mathbf{V}\leftarrow\mathbf{v}}$ be the function obtained by setting the variables in $\mathbf{V}$ to the values in $\mathbf{v}$. 
We have
\begin{eqnarray}
	\bigoplus_{\mathbf{v}\in\mathbb{F}_2^u} f_{\mathbf{V}\leftarrow\mathbf{v}}(K,{\rm IV})
	& = & \bigoplus_{\mathbf{v}\in\mathbb{F}_2^u} \left( V_{i_1}\cdots V_{i_u} p(K,{\rm IV}) \oplus q(K,{\rm IV}) \right)
	= p(K,{\rm IV}). \label{eqn-cube-reln}
\end{eqnarray}
Since each term of $q(K,{\rm IV})$ misses out at least one variable from $\mathbf{V}$, the sum over the ``cube'' $\mathbb{F}_2^u$ results in each term of
$q(K,{\rm IV})$ appearing an even number of times in the sum, and hence the total contribution of $q(K,{\rm IV})$ over the entire cube is 0. On the other hand
the monomial $V_{i_1}\cdots V_{i_u}$ is 1 for exactly one point in the cube $\mathbb{F}_2^u$ (i.e. when $\mathbf{v}=\mathbf{1}_u$). This shows that~\eqref{eqn-cube-reln}
indeed holds. The function $p(K,{\rm IV})$ is called a superpoly.

In an offline phase, the cryptanalyst obtains one or more superpolys. In the online phase, the cryptanalyst makes $2^u$ queries to the oracle for the stream cipher
fixing the variables in $V_1,\ldots,V_v$ outside that of $\mathbf{V}$ to some fixed values, and trying out all values of the variables in $\mathbf{V}$ over the
cube $\mathbb{F}_2^u$. From~\eqref{eqn-cube-reln}, the cryptanalyst obtains the value of $p(K,{\rm IV})$ for the unknown secret key $K$. If the form of the superpoly
is simple, then this value provides information to the cryptanalyst to narrow down the search for the secret $K$. 

Cube attacks were introduced in~\cite{DBLP:conf/eurocrypt/DinurS09}. Dynamic cube attacks were applied~\cite{DBLP:conf/fse/DinurS11a,DBLP:conf/asiacrypt/DinurGPSZ11}
to perform key recovery for Grain-128 (which is an earlier variant of Grain-128a). Cube attacks were improved in~\cite{DBLP:conf/crypto/TodoIHM17,DBLP:conf/crypto/WangHTLIM18}
using the division property. Refinement of this technique led to key recovery attacks on 190-round Grain-128a~\cite{DBLP:conf/eurocrypt/HaoLMT020}. 

Resistance to cube attacks arises from the effectiveness of the mixing of the secret key $K$ and the IV that is performed during the initialisation phase. 
After the initialisation phase, we may consider each of the bits of the state, i.e. the bits of $N$ and $L$, to be functions of $K$ and IV. The output bit is generated
from these state bits by applying the function $H$. So a design goal is to try to ensure that the functions determining the state bits after the initialisation phase
are complex nonlinear functions of $K$ and IV. In particular, such functions should not have low degrees. The functions determining the state bits after the
initialisation phase are based on the nature of the nonlinear functions $g$ and $h$ and the method of updating the state during initialisation. 

In our proposals, for both the 80-bit and the 128-bit security levels, the degrees of the functions $g$ and $h$ are higher than the corresponding functions
for Grain~v1 and Grain-128a. Further, the new method of updating the state during the initialisation (function $\sym{NSIG}$ in~\eqref{eqn-NSIG}) that we propose provides a 
better nonlinear mixing of the key and the IV compared to the state updation method during initialisation (function $\sym{NSI}$ in~\eqref{eqn-NSI}) used for all prior members of 
the Grain family (see the discussion following Remark~\ref{rem-new-init-tap-pos}). Due to the higher algebraic degrees of the constituent Boolean functions and
the better mixing of the key and the IV, we expect the new proposals to provide improved resistance to cube attacks.

\subsubsection{Conditional Differential Attacks \label{subsubsec-cond-diff-attack}}
The attacker introduces two IVs with a difference in a single bit. Consider the differences between the states generated from the two IVs. Initially there is a difference
at only a single point. As the initialisation proceeds, the difference propagates to more and more bits. Analysis of the stopping conditions for the differences leads to
two types of relations. The first type of relations involves the bits of the IV, while the second type of relations involves both the bits of the IV and the secret key. 
The first type of relations are easy to satisfy. For the second type of relations, the attacker tries to control the parameters such that the relations are expressed
as a direct sum of two Boolean functions, one on the key bits and the other on the IV bits. The ability to stop the propagation of the differences provides the adversary
with some information about the non-randomness of the keystream bits. This information can be exploited to launch an attack.
Conditional differential attack on Grain~v1 was introduced in~\cite{DBLP:conf/asiacrypt/KnellwolfMN10} and later explored in several papers for both Grain~v1 and
Grain-128a~\cite{DBLP:conf/cans/LehmannM12,DBLP:conf/acisp/Banik14,DBLP:journals/ccds/Banik16,DBLP:journals/iet-ifs/MaTQ17a,DBLP:journals/iet-ifs/MaTQ17,DBLP:journals/iet-ifs/LiG19}.

Resistance to conditional differential attacks is determined to a large extent by the algebraic degree of the update functions. It was noted in~\cite{DBLP:journals/ccds/Banik16}
that application of conditional differential attack to Grain~v1 is more difficult that to Grain-128a due to the degree of $g$ for Grain~v1 being higher
than the degree of $g$ for Grain-128a. For the new proposals at the 80-bit and the 128-bit security levels, the degrees of the function $g$ are higher than the corresponding 
degrees of the function $g$ for Grain~v1 and Grain-128a respectively.
Further, due to our new state update function during initialisation, the degrees of the functions representing the bits of the
LFSR grow much faster than the degrees of the functions obtained from the state update function during initialisation used for all earlier members of the Grain family. 
Due to the combined effect of these two points, the new proposals offer improved resistance to the conditional differential attacks than the previous proposals. 

\subsection{Near Collision Attacks \label{subsec-near-coll}}
Near collisions attacks on Grain~v1 were reported in~\cite{DBLP:conf/fse/ZhangLFL13,DBLP:conf/eurocrypt/ZhangXM18,DBLP:conf/acns/BanikCM23}. However,
as explained in~\cite{DBLP:conf/acns/BanikCM23} the attacks reported in~\cite{DBLP:conf/fse/ZhangLFL13,DBLP:conf/eurocrypt/ZhangXM18} are incorrect.
The idea behind the near collision attack on Grain~v1 described in~\cite{DBLP:conf/acns/BanikCM23} is the following. If about $2^{81.5}$ keystream bits
are generated from a single key and IV pair, then with high probability during the process two states will be encountered which differ only in one bit of the LFSR. 
This observation is used to build a state recovery algorithm which performs about $2^{74.6}$ encryptions and $2^{80.5}$ table accesses which leads to a time complexity
which is a little above the exhaustive search algorithm. 

While in principle the idea behind the attack can also apply to other members of the Grain family, the number of keystream bits needed to be generated from a
single key and IV pair to apply the attack appears to be well beyond the bound of $2^{64}$ that we have imposed. Further, as in the application
to Grain~v1, the time complexity of the attack also appears to be above that of exhaustive search.

\subsection{Time/Memory/Data Trade-Off Attacks \label{subsec-TMTO}}
Time/memory trade-off (TMTO) attacks were introduced by Hellman~\cite{DBLP:journals/tit/Hellman80} in the context of block ciphers. Application to stream ciphers
was pointed out by Babbage~\cite{Ba95} and Golic~\cite{DBLP:conf/eurocrypt/Goli97a} and is called the BG attack. Biryukov and Shamir~\cite{DBLP:conf/asiacrypt/BiryukovS00}
proposed a time/memory/data trade-off (TMDTO) attack called the BS attack. The A5/1 stream cipher was cryptanalysed using a trade-off attack in~\cite{DBLP:conf/fse/BiryukovSW00}.
Let $P$, $T$, $M$, and $D$ denote the pre-computation time, the online time, the memory, and the data required for carrying out a trade-off attack on a search
space of size $N_0$. In the context of stream ciphers $N_0$ is the size of the possible values of the state which is equal to $2^{\kappa_1+\kappa_2}$ for the
Grain family. The BG attack has the trade-off curve $TM=N_0$, $P=M$ and $T\leq D$ with best attack complexity given by the trade-off point
$T=M=D=N_0^{1/2}$; the BS attack has the trade-off curve $MT^2D^2=N_0^2$, $P=N_0/D$ and $T\geq D^2$ with the best attack complexity given by the trade-off point
$T=M=(N_0/D)^{2/3}$ and $P=N_0/D$. Application of TMDTO to stream ciphers with IV was proposed in~\cite{DBLP:conf/asiacrypt/HongS05}.
TMDTO has been applied to Grain-128a~\cite{DBLP:journals/dcc/DalaiPS22,DBLP:journals/tit/KumarS23}.

If $\kappa_1+\kappa_2=2\kappa$, then the complexity of the BG attack is not better than exhaustive search. The BS attack does not offer lower attack complexity. 
TMDTO attacks sometime allow $P$ to be greater than $2^{\kappa}$, for example the attack on Grain-128a in~\cite{DBLP:journals/tit/KumarS23} has $P=2^{170}$, which is significantly
greater than exhaustive search time of $2^{128}$. This aspect makes such attacks somewhat theoretical in nature. 
There is another problematic issue regarding the feasibility of TMDTO attacks. Such attacks require
very high memory, for example the attack on Grain-128a in~\cite{DBLP:journals/tit/KumarS23} require $M=2^{113}$. 

If $\kappa_1+\kappa_2<2\kappa$, then the BG attack requires $P=T=M=D=2^{(\kappa_1+\kappa_2)/2}<2^{\kappa}$. Since $T<2^\kappa$, such an attack may be considered to break the
security of the corresponding stream cipher. For example, for W-128, W-192, and W-256, the TMDTO attacks require all three of time (both pre-computation and online), memory, and 
data to be equal to $2^{120}$, $2^{176}$, and $2^{232}$ respectively. If we consider only time, then these are valid attacks. However, the memory requirements for these attacks 
are very high. 
In Appendix~\ref{app-bnd-keystream-mem}, we provide arguments which suggest that it is physically infeasible to apply attacks which require $2^{100}$ or
more bits of memory. In particular, the numbers $2^{176}$ and $2^{232}$ are both greater than the number of atoms in the world.
So even though the TMDTO attacks on W-128, W-192, and W-256 require time less than $2^{\kappa}$, the attacks cannot be mounted due to the infeasible memory requirement. 

More generally, theoretical time/memory/data trade-off attacks requiring physically infeasible amounts of memory are not a real concern.
For assessing the security of a real world cryptosystem, it is perhaps best to ignore physically infeasible attacks.


\subsection{Related Key Attacks\label{subsec-related-key}}
Related key attacks have been proposed on Grain~v1 and Grain-128a in~\cite{DBLP:conf/acisp/LeeJSH08} and later on Grain-128a 
in~\cite{DBLP:conf/acisp/BanikMST13,DBLP:journals/tifs/DingG13}. The attack in~\cite{DBLP:conf/acisp/BanikMST13} finds the secret key using a small multiple of
$2^{32}$ related keys and $2^{64}$ chosen IVs, whereas the attack in~\cite{DBLP:journals/tifs/DingG13} finds the secret key using two related keys, $2^{96}$ chosen
IVs, and $2^{104}$ keystream bits. The basic idea behind the attacks is that the states that are produced after the initialisation corresponding to two or more related keys
are related in some detectable manner. This suggests that the state update function that is used during initialisation does not adequately mix the secret key and the IV.
In our new proposals, we have introduced a new state update function during initialisation which performs a better mixing of the secret key and the IV (in a qualitative
sense). As a result, the new proposals will provide better resistance to related key attacks. We also note that whether related key attacks are a concern for
stream ciphers is a debatable issue.

\subsection{Fault Attacks\label{subsec-fault}}
Fault attacks on stream ciphers were proposed in~\cite{DBLP:conf/ches/HochS04}. A long and productive line of 
work~\cite{DBLP:conf/host/CastagnosBCDGGPS09,DBLP:conf/africacrypt/KarmakarC11,DBLP:conf/ches/BanikMS12,DBLP:conf/date/DeyCAM15,DBLP:conf/indocrypt/BanikMS12,DBLP:journals/tc/SarkarBM15,DBLP:journals/ccds/MazumdarMS15,DBLP:conf/space/SiddhantiSMC17,DBLP:journals/chinaf/LiLL24a,DBLP:journals/tcad/ChakrabortyMM17,DBLP:conf/space/BanikMS12a,DBLP:journals/access/SalamOXYPP21} has investigated various aspects of fault attacks on members of the Grain family. The design of the Grain family does not include any inherent protection
against fault attacks. In scenarios where fault attacks are considered a threat, appropriate counter-mechanisms (such as shielding) need to be deployed.

\section{State Guessing Attacks \label{sec-st-guess}}
For three of the proposals that we put forward, the length of the LFSR is smaller than the length of the NFSR. This opens up the possiblity of new kinds of attacks which
would not be applicable to previous members of the Grain family. In this section, we consider the possibility of applying a state guessing attack to proposals which
have the LFSR to be smaller than the NFSR.

At each time point, the state of the stream cipher consists of the state of the NFSR and the state of the LFSR. Suppose that the attacker knows a sufficiently long keystream segment. 
An attack approach is to guess the value of $\ell$ bits of the state and use the keystream to verify the guess. Here guessing essentially means trying out all possible 
$2^\ell$ values of the part of the state which is of interest. Suppose the verification requires $2^v$ operations. For the attack to be valid (i.e. faster than exhaustive
search), we must have $2^{\ell+v}<2^\kappa$. Since the NFSR is of length $\kappa_1\geq \kappa$, guessing the state of the NFSR does not provide a valid attack. 

For the W-series of stream ciphers, we propose the length of the LFSR to be $\kappa_2<\kappa$. This opens up the possibility that guessing the state of the LFSR
will lead to a valid attack. Suppose at some time point, the state $L^{(t)}$ of the LFSR is known. Since the evolution of the LFSR does not depend on the NFSR, given
$L^{(t)}$ we may assume that the state of the LFSR is known for all time points (greater than or less than $t$).
The output bit $\sym{OB}$ is given by the function $H$ which consists of $p=p_0+p_1$ bits from the NFSR and $q=q_0+q_1$ bits from the LFSR. Since the state of the LFSR
is known at all time points, the inputs to $H$ coming from the LFSR are fixed at all time points. So the keystream bit $z_t$ at time point $t$ can be expressed
as a sum of $p_1$ bits extracted from $N^{(t)}$, and the output of a nonlinear function $h_t$ on $p_0$ other bits, where for each $t$, $h_t$ is obtained
from $h$ by fixing the values of the $q_0$ variables obtained from the LFSR to specific values. Since the tap positions for the $p_1$ bits and the tap positions for the
$p_0$ bits are disjoint, it follows that $z_t$ is a balanced bit. So the value of $z_t$ reduces the number of options for the $p_0+p_1$ positions of $N^{(t)}$ to
$2^{p_0+p_1-1}$ from $2^{p_0+p_1}$. So effectively the knowledge of the bit $z_t$ reduces the number of possible choices for $N^{(t)}$ by half. Successive
keystream bits $z_{t+1},z_{t+2},\ldots$ further reduces the choices of $N^{(t)}$ by half each. This continues for a few time points until the nonlinear feedback
bit produced from $N^{(t)}$ gets shifted into one of the tap positions of the NFSR that is required to produce the output bit. From that time point, the keystream bit
becomes a more complex nonlinear function of the bits of $N^{(t)}$. To summarise, once $L^{(t)}$ is known at some time point $t$, any keystream bit at time point $t$,
or previous, or future time point can be written as a nonlinear function of the bits of $N^{(t)}$, where the nonlinear function varies with $t$. For a few time points
near $t$, the nonlinear function has a simple structure, but becomes more and more complex at time points further away from $t$. The requirement is to solve such 
a system of nonlinear equations in $\kappa_1$ variables to obtain $N^{(t)}$. Since the size of the LFSR is $\kappa_2$, to obtain a vaild attack it is required to solve the system
of nonlinear equations using less than $2^{\kappa-\kappa_2}$ operations. For the W-series proposals that we have put forward, the values of $\kappa-\kappa_2$
are 16, 32, 48 for the 128-bit, 192-bit and the 256-bit security levels, and the corresponding sizes of the NFSR are 128, 192, and 256 bits respectively.
We do not see how to solve the system of nonlinear equations in $\kappa_1$ variables in time less than $2^{\kappa-\kappa_2}$. Nonetheless, this is an attack avenue
which may be investigated futher.

\paragraph{Multi-target attacks.}
For the W-series of stream ciphers, since $\kappa_2$ is less than $\kappa$, a possible multi-target attack is the following. If the attacker obtains keystream segments
from $2^{\kappa_2}$ combinations of secret key and IV, then almost certainly for one of the targets, the state of the LFSR after initialisation, i.e. $L^{(0)}$
will be the all-zero value. As a result, the LFSR will play no part in the generation of the keystream and the entire keystream segment will be determined entirely by the NFSR.
This opens up the possibility that the attacker obtains the state of the NFSR for the corresponding target using less than $2^\kappa$ operations. The difficulty with this
approach is that the attacker does not actually know for which target the LFSR state is the all-zero value. We do not see any efficient way of determining this condition.
Also, $2^{\kappa_2}$ is quite high, and it is not practical for an attacker to obtain keystream segments from $2^{\kappa_2}$ combinations of key and IV.
Nevertheless, this is also an attack avenue which may need to be further studied. 

\section{Conclusion \label{sec-conclu} }
In this paper, we provided an abstract description of the Grain family of stream ciphers which formalised the different components of the family. 
We proposed strengthened definitions of the components which improve security. Seven new stream cipher proposals were proposed at various security
levels. We have argued that the new proposals provide good and improved resistance to known attacks on previous members of the Grain family. 
It is our hope that the present work will spur further research into the Grain family and more specifically into the security of the concrete proposals that 
we have put forward in this paper.

\section*{Acknowledgement} We would like to thank Alfred Menezes for helpful comments on Appendix~\ref{app-bnd-keystream-mem} and for pointing out~\cite{Sc20}.


\appendix
\section{Some Relevant Results on Correlation \label{app-corr}}
The generalised convolution theorem~\cite{DBLP:journals/tit/GuptaS05a} expresses the Walsh transform of a finite sum of Boolean functions
in terms of the Walsh transforms of the individual Boolean functions. Below we restate the result in terms of correlations.
\begin{theorem}[Theorem~2 of~\cite{DBLP:journals/tit/GuptaS05a}]\label{thm-gen-conv}
	Let $h_1,\ldots,h_k$ be $n$-variable Boolean functions and $h(\mathbf{X})=h_1(\mathbf{X})\oplus \cdots\oplus h_k(\mathbf{X})$. Then for all $\mathbf{u}\in \mathbb{F}_2^n$,
	\begin{eqnarray}\label{eqn-gen-conv}
		\sym{corr}_h(\mathbf{u}) & = & \sum_{\{\mathbf{u}_i\in\mathbb{F}_2^n: 1\leq i\leq k-1\}} \prod_{i=1}^k \sym{corr}_{h_i}(\mathbf{u}_i\oplus \mathbf{u}_{i-1}),
	\end{eqnarray}
	where $\mathbf{u}_0=\mathbf{0}_n$ and $\mathbf{u}_k=\mathbf{u}$. 
\end{theorem}
The following simple results are easy to verify.
\begin{proposition}\label{prop-corr-degenerate}
	Let $\mathbf{X}=(X_1,\ldots,X_m)$, $\mathbf{Y}=(Y_1,\ldots,Y_n)$ and $f(\mathbf{X},\mathbf{Y})=g(\mathbf{Y})$, i.e. $g$ is degenerate on $\mathbf{X}$.
	Then $\sym{corr}_f(\mathbf{0}_m,\bm{\beta})=\sym{corr}_g(\bm{\beta})$, and $\sym{corr}_f(\bm{\alpha},\bm{\beta})=0$ for $\bm{\alpha}\neq \mathbf{0}_m$. 
\end{proposition}
\begin{proposition}\label{prop-corr-perm}
	Let $f(X_1,\ldots,X_n)=h(\psi(X_1,\ldots,X_n))$, where $\psi$ is a bit permutation. Then $\sym{corr}_f(\bm{\alpha})=\sym{corr}_h(\psi(\bm{\alpha}))$.
\end{proposition}
\begin{proposition}\label{prop-dsum-corr}
	Let $f(X_1,\ldots,X_{n_1},Y_1,\ldots,Y_{n_2}) = g(X_1,\ldots,X_{n_1}) \oplus h(Y_1,\ldots,Y_{n_2})$. Then
	$\sym{corr}_f(\bm{\alpha},\bm{\beta}) = \sym{corr}_g(\bm{\alpha}) \sym{corr}_h(\bm{\beta})$.
\end{proposition}

\section{Bounds on Number of Keystream bits and Physical Memory \label{app-bnd-keystream-mem}}
In this section, we provide (informal) justification for setting bounds on the number of keystream bits that should be generated from a single key and IV pair,
and the number of bits of memory which may be considered physically possible.

\paragraph{Bound on the number of keystream bits generated from a single key and IV pair.}
Recall that we have set the maximum number of keystream bits that can be generated from a single key and IV pair to be $2^{64}$. This may be contrasted
with the standardised AES-GCM which permits $2^{32}-2$ blocks (i.e. less than $2^{39}$ bits, since a block consists of 128 bits) to be encrypted/authenticated with a single
key and IV pair.

To put the figure of $2^{64}$ in a realistic context, we note that the storage capacity of CERN crossed one exabyte (i.e. $2^3\cdot 10^{18}\approx 2^{62.8}$ bits) in September 
2023\footnote{\url{https://home.cern/news/news/computing/exabyte-disk-storage-cern}, (accessed on 7th October, 2025)}. This data is stored at multiple data centres
across different countries. It is completely impractical to encrypt the whole CERN storage using a single key and IV pair. One good reason for the impracticality is that 
the keystream generation does not support random access, i.e. it is not possible to generate later portions of the keystream without generating the earlier portions. So a 
single monolithic encryption of a huge number of bits using a single key and IV pair will create tremendous inefficiencies in encrypting and decrypting portions that 
are not close to the start of the message. One of the advantages of using an IV is to avoid such inefficiencies. Using the same secret key, different IVs can be used to 
encrypt (and decrypt) disjoint reasonable size segments of the message.

%
%

\paragraph{Bound on the memory.}
The number of atoms in the world\footnote{\url{https://en.wikipedia.org/wiki/Atom}, accessed on 7th October, 2025} is about $1.33\cdot 10^{50}\approx 2^{166.5}$.
So any attack requiring memory around or more than this figure is clearly physically impossible. We next provide two informal arguments as to why even
memory around $2^{100}$ or more bits is physically impossible.

The first argument goes as follows. As mentioned above, the storage capacity of CERN crossed one exabyte in 2023. It is reported that Google has 10 to 15 exabytes
of storage, and AWS possibly has more storage. Suppose we consider 100 exabytes (i.e. $100\cdot 2^3\cdot 10^{18}\approx 2^{69.4}$ bits) of storage. 
The question is how much physical space is required to store this amount of data? CERN data, as well as Google and AWS data are stored at multiple data centres across the world, 
making it difficult to estimate the physical area required to store the entire $2^{62.8}$ bits of data. For the sake of concreteness, we make the ad-hoc
assumption that the 
entire 100 exabyte data can be physically stored in a building whose surface area is equal to the size of a FIFA football field which is a 105m $\times$ 68m 
rectangle\footnote{\url{https://publications.fifa.com/de/football-stadiums-guidelines/technical-guideline/stadium-guidelines/pitch-dimensions-and-surrounding-areas/},
accessed on 7th October, 2025}.
The surface area of Antartica is about 14.2 million square kms\footnote{\url{https://en.wikipedia.org/wiki/Antarctica}, accessed on 7th October, 2025}, and so about
$2^{30.9}$ buildings whose surface area is equal to that of a FIFA football field will cover Antartica. If each of these buildings houses $2^{69.4}$
bits of data, then the total amount of data that can be stored in all of these buildings is about $2^{100.3}$ bits. 
It is clearly absurd to cover an area as large as Antartica in such a manner.

The second argument proceeds somewhat differently and is based on a similar argument appearing in Pages 13 to 15 of~\cite{Sc20}. A micro SD card is of size
15mm~$\times$~11mm$\times$~1mm. Let us suppose that a micro SD card can store 1 terabyte (i.e. $2^3\cdot 10^{12}\approx 2^{42.9}$ bits). If we tile entire
Antartica using micro SD cards, then the number of such cards will be about $2^{56.3}$. The total bits that can be stored in all of these cards is at most
$2^{99.2}$ bits.
It is clearly absurd to cover an area as large as Antartica with micro SD cards.

Both of the above arguments are based on the limitation of physical size and indicate that memory of $2^{100}$ or more bits is physically infeasible.


\end{document}